\documentclass[british]{article}
\usepackage[utf8]{inputenc}

\usepackage{graphicx}
\usepackage{amscd}
\usepackage{amsmath}
\usepackage{amsfonts}
\usepackage{amssymb}
\usepackage{bm} 
\usepackage{setspace}
\usepackage{enumerate}         
\usepackage{color}
\usepackage{xcolor} 
\usepackage{url}
\usepackage{amsthm}
\usepackage{hyperref}
\usepackage{bm}
\usepackage{xfrac}
\usepackage{centernot}
\usepackage{mathtools} 
\usepackage{nicefrac} 
\usepackage{pdflscape,afterpage} 
\usepackage{rotating}
\usepackage{comment}

\usepackage[T1]{fontenc}
\DeclareFontFamily{T1}{calligra}{}
\DeclareFontShape{T1}{calligra}{m}{n}{<->s*[1.1]callig15}{}
\DeclareMathAlphabet\mathcalligra   {T1}{calligra} {m} {n}

\usepackage{algorithm}
\usepackage[noend]{algpseudocode}
\usepackage{xpatch} 

\makeatletter
\xpatchcmd{\algorithmic}{\itemsep\z@}{\itemsep=1ex plus2pt}{}{}
\makeatother

\usepackage[british]{babel}
\usepackage{hyphenat} 
\usepackage{tikz}
\usetikzlibrary{arrows.meta, decorations.text, decorations.pathmorphing,decorations.pathreplacing}

\usepackage{float} 

\numberwithin{equation}{section}

\theoremstyle{plain}
\newtheorem{theorem}{Theorem}[section]

\newtheorem{corollary}[theorem]{Corollary}

\newtheorem{proposition}[theorem]{Proposition}

\newtheorem{definition}[theorem]{Definition}

\theoremstyle{remark}
\newtheorem{remark}[theorem]{Remark}
\newtheorem{example}[theorem]{Example}

\newcommand{\Xt}{(X_t)_{t\geq0}}

\newcommand{\ii}{{I \setminus \{i\}}}
\newcommand{\norm}[1]{\left\Vert#1\right\Vert}

\newcommand{\RR}{\mathbb{R}}
\newcommand{\QQ}{\mathbb{Q}}
\newcommand{\PP}{\mathbb{P}}
\newcommand{\CC}{\mathbb{C}}
\newcommand{\NN}{\mathbb{N}}
\newcommand{\EE}{\mathbb{E}}
\newcommand{\FF}{\mathbb{F}}

\newcommand{\cC}{\mathcal{C}}

\newcommand{\cF}{\mathcal{F}}
\newcommand{\cI}{\mathcal{I}}

\newcommand{\cM}{\mathcal{M}}

\newcommand{\cU}{\mathcal{U}}

\newcommand{\pd}[2]{\frac{\partial #1}{\partial #2}}

\newcommand{\scal}[2]{\left\langle{#1},{#2}\right\rangle}

\renewcommand{\Re}{\mathrm{Re}\,}
\renewcommand{\Im}{\mathrm{Im}\,}

\newcommand{\G}{\Gamma}
\newcommand{\Om}{\Omega}

\newcommand{\epsi}{\varepsilon}
\newcommand{\define}{\vcentcolon =}
\renewcommand{\ii}{\mathrm{i}\mkern1mu}
\newcommand{\nn}{\emph{NN}}

\newcommand{\Inc}{\textrm{Inc}}
\newcommand{\phitilde}{\widetilde{\phi}}
\newcommand{\psitilde}{\widetilde{\psi}}

\DeclareMathOperator{\spn}{span}

\newcommand{\Esign}[1]{\mathbb{E}\left[ #1 \right]}  
\newcommand{\Econd}[2]{\mathbb{E}\left[\left.#1\right|#2\right]}  

\begin{document}
\title{Consistent Recalibration Models and Deep Calibration}
\author{Matteo Gambara\footnote{matteo.gambara@math.ethz.ch, ETH
    Zürich, Rämistrasse 101, Zürich.} \,\,\,\&\, Josef
  Teichmann\footnote{josef.teichmann@math.ethz.ch, ETH Zürich,
    Rämistrasse 101, Zürich.}}
\date{5\textsuperscript{th} May 2021}
\maketitle
\begin{abstract}
  Consistent Recalibration (CRC) models have been introduced to
  capture in necessary generality the dynamic features of term
  structures of derivatives' prices. Several approaches have been
  suggested to tackle this problem, but all of them, including CRC
  models, suffered from numerical intractabilities mainly due to the
  presence of complicated drift terms or consistency conditions. We
  overcome this problem by machine learning techniques, which allow to
  store the crucial drift term's information in neural network type
  functions. This yields first time dynamic term structure models
  which can be efficiently simulated.
\end{abstract}
	
\selectlanguage{british}\allowhyphens
	
\tableofcontents
	
\section{Introduction}
	
Term structures of prices exist in many different markets and belong
to the most challenging topics for dynamic modelling in mathematical
finance or econometrics. Reasons for this are two-fold: first, we have
to deal with potentially infinitely many, strongly dependent prices
satisfying mutual relations, and, second, the dynamics of those prices
is subject to absence of arbitrage conditions. Whence it is a delicate
issue to write models which take the full market information as state
variables into account, i.e.~all available prices, and which evolve at
the same time in a way which satisfies all constraints. The problem
has been successfully dealt with in the context of bond prices within
the Heath-Jarrow-Morton framework \cite{hjm92}, but already the term
structure of plain vanilla option prices on one underlying $(S_t)$ has
so far been too challenging to come up with fully satisfying solutions
despite deep theoretical insights. Of course there are more involved
term structures stemming from volatility cubes in treasury or the
combined S\&P -- VIX term structures. We consider our work a first
step in the direction of tractable term structure modelling.
	
Let us be more precise on this point (we assume an environment free of
interest rates): the \emph{classical approach} to modelling term
structures consists in choosing a class of stochastic processes $(S_t)$,
which models the price of the underlying with respect to a pricing or
physical measure, on a stochastic basis, which encodes the information
structure through a filtration. Given market data one model is
selected from the given class via calibration, i.e.~the procedure
guaranteeing that all market prices are reproduced by the model. Then
new products are priced and hedged with the calibrated model. In the
realm of the term structure of option prices on one underlying,
parametric models like SABR, Heston or rough volatility models are
used for this purpose; similarly, we can use non-parametric models like local volatility
models which are stochastic local volatility models. This works in many
respects very well, but suffers from dynamic shortcomings, i.e. newly
arriving information leads via recalibration to a new model choice,
whence an inconsistency over time in modelling (see, for example,
\cite{Dupire96}, \cite{Gatheral18}, \cite{smile11}, \cite{Zubelli19}).
		
Term structure models try to overcome this issue by making market
prices state variables of the model: the price to pay for this
neo-classical approach is complexity. We shall outline without going
into detail some of the suggested approaches here. Let us denote by
${(C_t(T,K))}_{0 \leq t \leq T}$ the stochastic process of plain
vanilla options on one underlying $ (S_t) $. Here $ K $ denotes the
option's strike price and $ T $ its maturity; let us ignore interest (and dividend) rates for the moment. The dynamic and static
no arbitrage constraints are expressed on the given stochastic basis
by the existence of an equivalent measure $ \mathbb{Q} $ such that
\[
  E_{\mathbb{Q}}[(S_T-K)_+|\mathcal{F}_t] = C_t(T,K)
\]
and
\[
  C_t(T,0) = S_t
\]
for $ 0 \leq t \leq T $ hold true. The question has been raised which
\emph{codebook} should be used to facilitate the most in dealing with
those constraints. Roughly speaking three suggestions have been made,
which we shortly introduce here:
\begin{itemize}
\item Given the current market price $S_t$ there is a unique (implied)
  volatility $ \sigma_t(T,K) $ such that the Black-Scholes formula
  $ \operatorname{BS} $ produces the correct market price
  \[
    \operatorname{BS}(T,K,S_t,\sigma_t(T,K)) = C_t(T,K)
  \]
  for $ 0 \leq t \leq T $. This, however, yields two problematic
  aspects: first, how to deal with the dynamic absence of arbitrage,
  and, second, how to express the static absence of arbitrage
  conditions for implied volatilities. Under some regularity
  assumptions it turns out that one can write necessary conditions for
  such a dynamics by imposing that
  $ (\sigma_t(T,K))_{0 \leq t \leq T} $ remains within the set of
  statically arbitrage free surfaces and that the process
  \[
    \big(\!\operatorname{BS}(T,K,S_t,\sigma_t(T,K))\big)_{0 \leq t \leq T}
  \] is a martingale for all $ T, K $ provided that $ S $ is one. This
  case was outlined by Schweizer and Wissel in \cite{schweizer08}
  following the work done by Sch\"onbucher in \cite{schonbucher99}.
	
\item Given the current market price $S_t$ there is a unique local
  volatility $ \sigma(t,s) $ such that the pricing operator
  $ \mathcal{P} $ for plain vanilla calls of the local volatility
  equation
  \[
    d X_r = \sigma(r,X_r) \,dB_r \, ; \; X_t = S_t
  \]
  produces the market prices $ C_t(T,K) $ for $ T \geq t $ and all
  $K$. Here situation is simpler, since local variance must simply be
  non-negative, however, the dynamic absence of arbitrage is more
  involved: first, we impose that $ S $ has to be a martingale, and,
  second, we need that
  \[
    \big(\mathcal{P}(S_t,\sigma_t(\cdot\,,\cdot),T,K)\big)_{0 \leq t \leq T}
  \]
  is a martingale, too, for all $ T , K $. This involves a full
  fledged solution operator of local volatility equations (see
  \cite{Carmona2012} for more information).\\
  In a similar context, Wissel builds on top of the cited reference in \cite{wissel07} 
  on a discrete expiry-strike grid an analogous codebook addressing additionally 
  the relevant problem of existence and uniqueness of a solution for the price SDE.
	
\item Given the current market price $S_t$ there is a unique
  time-dependent L\'evy process $(b_t,c_t,\nu_t) = \mathcal{L}_t $
  characterised by its L\'evy triplet such that the pricing operator
  $ \mathcal{P}$ of the corresponding exponential L\'evy
  \emph{martingale} produces the market prices $ C_t(T,K) $ for
  $ T \geq t $ and all $K$, just as before. Here the situation is slightly
  more complicated than in the case of local volatility, since L\'evy
  triplets are more complicated objects than non-negative functions of
  two variables. On the other hand, the pricing operator is
  considerably simpler due to Fourier pricing. We impose again that
  $ S $ is a martingale, and that
  \[
    \big(\mathcal{P}(S_t,\mathcal{L}_t,T,K)\big)_{0 \leq t \leq T}
  \]
  is a martingale, too, for all $ T , K $. This approach has been
  developed at the same time by Carmona and Nadtochiy
  (\cite{Carmona2012}) and by Kallsen and Kr\"uhner in
  \cite{Kallsen2015}, still with different choices of
  codebooks. Therefore, in the following we will refer to this
  approach as CNKK deploying author's initials.

\end{itemize}
	
It is the goal of this article to make the CNKK approach work in the
setting of Consistent Recalibration Models, i.e.~where we consider
tangent affine models. Basically speaking this amounts to storing the
information of a non-linear drift operator in a neural network in an
optimal way, in case when the time evolution is locally mimicking a
dynamically changing affine model. In our view, this is the simplest way to consistently
construct term structure dynamics, which do not come from finite dimensional realisations.

Actually this information, which is stored in the drift, corresponds
to solving an inverse problem or a calibration problem, see
\cite{CucKhoTei:2020} and the references therein for a general
background of this problem: more precisely, it is the inversion of the
above mentioned pricing operators given the current market state of
the underlyings' price and the term structure of derivatives'
prices. Let us outline this in case of the L\'evy codebook: there the
inverse problem corresponds to calculating the time-dependent L\'evy
triplet $L_0$ given the price of the underlying $S_0$ and the term
structure of derivatives' prices. Even though the map from model
characteristics to prices is usually smooth, it is due to
\emph{smoothing} properties, often hard to invert: existence,
uniqueness and stability issues (in the sense of Jacques Hadamard\footnote{The
three conditions Hadamard established in the 20\textsuperscript{th} century 
to define well-posedness of a problem 
$X\supseteq U\ni x \mapsto F(x) = y \in Y$ are existence of a solution, i.e. $\forall \,x\in U, \exists \,y \in Y$; 
uniqueness of the solution; 
and continuity with respect to the input, that is, if we are 
in metric spaces with distances $d_X$ and $d_Y$, 
$\forall \,\epsi > 0, \exists \, \delta_\epsi > 0$ such that 
$\forall \,x_1, x_2 \in U$ and $y_1, y_2 \in Y$ for which 
we have $y_1=F(x_1)$ and $y_2=F(x_2)$ we must have 
$d_X(x_1,x_2) < \delta_\epsi \Rightarrow d_Y(y_1,y_2)< \epsi$. 
If one (or more) of these conditions is not fulfilled, a problem is
said to be \emph{ill-posed}.})
appear. Machine learning technology provides one way to fix these
issues, which otherwise require sophisticated regularization
techniques, by implicit regularization, see, e.g.
\cite{HeiTeiWut:2020}. In other words: learning the map from
derivatives prices (given the current price of the underlying) to
model characteristics $L_0$ and storing the information in a neural
network provides an accurate map satisfying Hadamard's
requirements. We shall pursue this approach not in the originally
proposed way by solving a supervised learning problem, see, e.g.
\cite{Hernandez1}, but rather by storing first the information of the
pricing operator in a neural network and then inverting this network,
compare here, e.g. \cite{Horvath2019}.\\
In the very same years, other applications of ML to the same problem 
were developed. One of these is \cite{wiese2019}, where Wiese and co-authors 
make use of a Generative Adversarial Networks (GANs) to learn to 
simulate time series of the codebook introduced by Wissel 
in \cite{wissel07}. Thanks to the particular choice of the codebook,
it is relatively easy to learn new plausible prices that satisfy
static arbitrage constraints on discrete grids: as we recalled above, 
non-negativity of the local volatility process is the only requirement 
in this case. On the other hand, the work we present here does not try to get 
rid of static arbitrage possibilities alone, but of dynamic arbitrage 
as well, thus answering in a more complete and 
satisfactory way to the need of a realistic equity option market 
simulator, which was the \emph{raison d'\^etre} of \cite{wiese2019}.

One could also view our current model as an unusually parameterized
neural stochastic differential equation (NSDE) model, see, e.g.
\cite{CucKhoTei:2020} for details on this concept. NSDEs,
i.e.~stochastic differential equations with neural network
characteristics, are a wonderful concept to construct non-parametric
models, but it is quite delicate to write constraint dynamics with
neural networks. Therefore we have chosen Consistent Recalibration
Models with tangent affine models, where it is easier to express
constraints in terms of neural networks for the drift. A
non-parametric approach alluding to NSDEs (and based on random
signature methods, see \cite{CucGonGriOrtTei:2020}) will be presented
in upcoming work.

The remainder of the article is structured as follows. In the second
section, we introduce the concept of a L\'evy triplet \emph{codebook},
as shortly alluded to above, and we define consistent recalibration
(CRC) models. We also briefly review affine models and embed
stochastic volatility affine models in the context of CRC models,
outlining some key properties of this codebook.  The third section is
dedicated to one of the building blocks of the whole theory:
generalised Hull-White extensions.  We do not simply use the
Hull-White extension, as it was exploited when first defined, for the
calibration at the initial time of the term structure in the interest
rate models, but we think of it as a tool that allows for
recalibration of the model parameters. Further, we talk about a
\emph{generalised} extension, since we are replacing the pure drift
addition typical of interest rate models with a L\'evy process, which
naturally encodes a greater calibration power.  To make things
clearer, an example is laid out in the fourth section, where we
analyse how the generalised Hull-White extension is added to the
Heston model in order to get a consistent recalibration model, what we
call a generalised Bates model. The same example is important because
a very similar version of this model has been implemented numerically.
The fifth section is devoted to the CNKK equation: how it is derived,
defined and how it can be seen as a generalisation of the HJM
equation. Section~\ref{sec:CRC_math} is dedicated to the formal
definition of CRC models of stochastic volatility affine models with
piecewise constant model parameters for pricing stocks' derivatives.
Some numerical considerations are also listed, to show what are the
most relevant aspects to deal with in case of a concrete
implementation.  One of this point is the main subject of
Section~\ref{sec:deepCal}, where we discuss about calibrating the
model using a neural network. 
In the end, the
conclusion summarises the main results and novelties of the paper and
an appendix collects some representative pictures of volatility surfaces
obtained with our model.

\bigskip \textbf{Notation.} The set $\NN_0$ denotes the set of natural
numbers with $0$ included; $\RR_+^m$ the real vectors in $\RR^m$ whose
components are greater than $0$. Analogously, $\RR_{\geq0}$ denotes 
the real semi-line starting from 0 (included).
	
With $L(X)$ we denote the set of $X$-integrable predictable processes
for a semimartingale $X$. If we talk about $(X,Y)$ as a
$(m+n)$-semimartingale, we mean that $X$ is an $\RR^m$-valued
semimartingale and $Y$ is an $\RR^n$-valued semimartingale.
	
Whenever we apply complex logarithm of continuous functions
$\RR^n \ni u\mapsto f(u)\neq 0$, we use the normalisation
$\log f(0) = 0$, so that logarithms are uniquely defined (see Proposition 2.4 in \cite{keller-ressel2013}).
	
For the sake of readability, SDE and SPDE will be used as acronyms for
stochastic differential equation and stochastic partial differential
equation, respectively.
	
\section{Consistent Recalibration Models}
	
We take inspiration from the seminal paper of Kallsen and Kr\"{u}hner
\cite{Kallsen2015}, but we place ourselves in a more general framework that does
not necessarily rely on the infinite divisibility of the
processes. Let $(\Om, \cF, (\cF)_{t \geq 0}, \QQ)$ denote a filtered
probability space, where $\QQ$ represents a risk-neutral measure so
that discounted asset prices are supposed to be $\QQ$-local martingales. All
expectations, if not differently specified, are also taken with
respect to the same probability measure and are denoted by $\EE$. We
consider an adapted (multivariate) stochastic process $X \define \Xt$
taking values in $\RR^n$, which can be considered as logarithms of
price processes.
	
We assume that call options of any strike and maturity are liquidly
traded and denote the time $t$ value of a call option with maturity
$T$ and strike $K$ by $C_t(T,K)$. Having liquid market prices
for all maturities and strikes\footnote{In practice, and for our ensuing numerical algorithm, this is never the case.} translates, in mathematical terms, in having the marginal
distributions of several (at most $n$) underlying processes (under the pricing measure). Similar
to \cite{Kallsen2015}, we could at this point require that the given
marginal distributions of $X$ are infinitely divisible, that the
characteristic functions are absolutely continuous with respect to
time 
and define the \emph{codebook} of our model as a ``forward'' L\'{e}vy
exponent and then define dynamics for such (infinite dimensional)
codebook.  But rather than focusing on the joint behaviour of $X$ and the
codebook, whose dynamics are expressed in function of a generic
semimartingale $M$, as done in \cite{Kallsen2015}, we exploit the
intuition behind the choice of such codebook, since it provides easier
conditions to avoid dynamic and static arbitrage, but we build around
it a new framework.
	
For this reason, we conveniently define consistent recalibration
models as models that keep \emph{consistency}, which means that future
realisations will be in a neighbourhood of the current state that can
always be reached with positive probability, and that are
\emph{analytically tractable}, thus looking as finite factor models
instantaneously. This is reached by introducing stochastic parameters,
whose dynamics could be extrapolated by market data, and by means of a
Hull-White extension, used to compensate the stochastic updates in the
parameters, while leaving the marginal distributions of the state
variables unchanged.
	
We start defining the set of functions that is the base for our
theory:
\begin{definition}[$\G_n$]
  The set $\G_n$ denotes the collection of continuous functions
  $\eta : \RR^n \times \RR_{\geq0} \to \CC$ such that there exists a
  càdlàg process $Z$ with independent increments and finite
  exponential moments $\EE[\exp((1+\epsi)\|Z_T\|)] < \infty$ for all
  $T\geq 0$ and for some $\epsi > 0$ satisfying
  \begin{equation}\label{forwardCharGn}
    \EE[\exp(\ii\scal{u}{Z_T}) ] = \exp\left(\ii\scal{u}{Z_0} + \int_0^T \eta(u,r)\,dr\right)
  \end{equation}
  for $u \in \RR^n$.
\end{definition}
	
	\begin{remark}\label{remark:strip}
          Requiring that for some $\epsi > 0$,
          $\EE[\exp((1+\epsi)\|Z_T\|)] < \infty$ implies that we can
          extend the left hand side of \eqref{forwardCharGn} to the
          strip $-\ii[0,1]^n \times \RR^n$, thus we could choose, for
          example, $u=-\ii$.
	\end{remark}
	
	\begin{remark}
          All functions $\eta \in \G_n$ are necessarily of
          L\'evy-Khintchine type at the short end ($r=0$ in \eqref{forwardCharGn}). Note that
          by doing so, we are extending the definitions given in
          \cite{Carmona2012} and \cite{Kallsen2015} since we only
          assume the function $\eta$ to be of L\'evy-Khintchine form
          at the short end.
	\end{remark}
	
	\begin{remark}\label{remarkXi}
          Often, elements in $\G_n$ are subject to additional
          no-arbitrage constraints in order to satisfy, for example,
          the martingale property for both the price processes
          $S = \exp(X)$ and call options $C_t(T,K)$ for all $T,K >0$
          (see Theorem 3.7 in \cite{Kallsen2015}).  For instance, if
          we consider $X = (X^i)_{i=1}^d$ as being a log-price process
          for some $i$, we also assume $\exp(X^i)$ is a martingale,
          which is equivalent to state that $\eta(e_i, r)=0$ with $e_i$
          being the $i$-th basis vector of $\RR^n$,
          i.e.~$ \scal{e_i}{X_T} = X^i_T $. We assume tacitly that
          such conditions are imposed if necessary. Notice in case of
          a components of $ X $ corresponding to interest rates we do
          not need to impose such a condition (since we do not need
          martingality).
	\end{remark}

	We can think of the set $\G_n$ as a chart, in the language of
        geometry, or codebook, in the language of mathematical
        finance, for all liquid market prices at one instant of
        time. If we want to consider their time evolution, we had
        better define $\G_n$-valued processes:
	\begin{definition}[$\G_n$-valued semimartingale]
          Let $(\Om,\cF,{(\cF_t)}_{t \geq 0}, \QQ)$ be a filtered
          probability space. A stochastic process $ \eta $ is called a
          \emph{$\G_n$-valued semimartingale} if
          $ {(\eta_t(u,T))}_{0 \leq t \leq T} $ is a complex-valued
          semimartingale for $ T \geq 0 $ and $ u \in \RR^n $ and if
		$$ 
		\bigl((u,r)\mapsto \eta_t(u, r+t) \bigr) \in \Gamma_n
                \,.
		$$ 
		In particular, all trajectories are assumed to be
                c\`adl\`ag.
              \end{definition}

	\begin{definition}[Regular decomposition]\label{def:regDecomposition}
          We say that $ \eta $ allows for a \emph{regular
            decomposition} with respect to a $d$-dimensional
          semimartingale $M$ if there exist predictable processes
          $ {(\alpha_t(u,T))}_{0 \leq t \leq T} $ taking values in
          $ \CC $ with $\alpha_t(0,T)=0$ for all $0\leq t\leq T$ and
          $ {(\beta_t(u,T))}_{0 \leq t \leq T}$, $\CC^d$-valued, with
          $\beta_t^i(0,T)=0$ for all $i$ and all $0\leq t\leq T$ and
          for $ T \geq 0 $ and $ u \in \RR^n $ such that
          \begin{equation}\label{regularDec}
            \eta_t(u,T) = \eta_0(u,T) + \int_0^t \alpha_s(u,T) \, ds + \sum_{i=1}^d \int_0^t \beta^i_s(u,T) \, dM^i_s
          \end{equation}
          for $0\leq t\leq T$, and
          $ {\left(\sqrt{\int_t^T {\norm{\beta_t(u,r)}}^2 \,
                dr}\right)}_{t \geq 0} \in L(M) $.
	\end{definition}

	In view of the two new definitions, we can also generalise the
        condition expressed in \eqref{forwardCharGn}:
	\begin{definition}[Conditional expectation condition]
          Let $(\Om,\cF,{(\cF_t)}_{t \geq 0},\QQ)$ be a filtered
          probability space, then we say that a tuple $ (X,\eta) $ of
          an $n$-dimensional semimartingale $ X $ and of a
          $ \G_n $-valued semimartingale $ \eta $ satisfies the
          \emph{conditional expectation condition} if
          \begin{equation}\label{cec}
            \Econd{\exp(\ii \scal{u}{X_t})}{\cF_s} = \exp \left(\ii \scal{u}{X_s} + \int_s^t \eta_s(u,r) \, dr \right)
          \end{equation}
          for $ 0 \leq s \leq t $.
	\end{definition}
	
	At this point, the link of the whole theory to its
        discrete-time counterpart exposed in \cite{Richter2017}
        becomes even clearer:
	\begin{definition}[Forward and process characteristics]
          Let $X$ be an adapted semimartingale taking values in
          $\RR^n$ and $\eta$ a $\G_n$-valued semimartingale with
          $\eta_s(0,t)=0$ for all $0\leq s\leq t$ and for which the
          conditional expectation condition \eqref{cec} is
          satisfied. Then the process $\eta$ is called \emph{forward
            characteristic process of $X$}.  Analogously, the process
          denoted by $\kappa_s^X$ and that coincides with the short
          end of the forward characteristics of $X$,
          i.e.~$\eta_{s-}(\cdot, s)$, with $\kappa_s^X(0)=0$ for all
          $s\geq 0$ is said \emph{(process) characteristic of $X$}.
	\end{definition}
	
	\begin{remark}
          Note that both processes are uniquely defined (up to a
          $d\QQ\otimes dt$-nullset):
          \begin{enumerate}
          \item The normalisation $\eta_s(0,t)=0$ for all
            $0\leq s\leq t$ ensures that the map
            $u \mapsto \eta_s(u,t)$ is continuous and uniquely defined
            through the use of the complex logarithm.
          \item Since the adapted process
            $\left(\exp\left(\ii\scal{u}{X_s} - \int_0^s
                \kappa_r^X(u)\, dr \right)\right)_{s\geq 0}$ is a
            local martingale (see Theorem \ref{cecImpliesDC} below)
            and $\kappa_r^X(0)=0$ for any $r\geq 0$, uniqueness
            follows from Lemma A.5 in \cite{Kallsen2015}.
          \end{enumerate}
	\end{remark}
	
	\begin{definition}[Term structure for derivatives]
          We call the tuple $ (X,\eta) $ of an $n$-dimensional
          semimartingale $ X $ and of its $ \Gamma_n $-valued forward
          characteristic process $\eta$ a \emph{term structure model
            for derivatives' prices}.
	\end{definition}

	With the following theorems, we are able to characterise which
        processes $\eta$ can be considered forward processes, given
        the existence of a regular decomposition. Discrete versions of
        the same theorems are given in \cite{Richter2017}, while
        proofs for the continuous case under examination are in
        \cite{Kallsen2015}.
	
	\begin{theorem}\label{cecImpliesDC}
          Let $(\Om,\cF,{(\cF_t)}_{t \geq 0},\QQ)$ be a filtered
          probability space together with a tuple $ (X,\eta) $ of an
          $n$-dimensional semimartingale $ X $ and of a
          $ \Gamma_n $-valued semimartingale $ \eta $ satisfying the
          \emph{conditional expectation condition}, then
          \begin{itemize}
          \item the differentiable, predictable characteristic
            $ \kappa^X $ of the $n$-dimensional semimartingale $ X $
            exists and is given by $ \kappa_t^X(u)= \eta_{t-}(u,t) $
            (usually called \emph{short end condition}) for
            $ t \geq 0 $ and $ u \in \RR^n$, i.e. the process
            \begin{equation}\label{predictable-char}
              \exp \left(\ii \scal{u}{X_t} - \int_0^t \eta_{s-}(u,s) ds \right)
            \end{equation}
            is a local martingale.
          \item If $ \eta $ allows for a regular decomposition
            \eqref{regularDec} with respect to a $d$-dimensional
            semimartingale $ M $, then the \emph{drift condition}
            \begin{equation}\label{drift-cond}
              \int _t^ T \alpha_t(u,r) \, dr = \eta_{t-}(u,t) - \kappa_t^{(X,M)}\left(u, - \ii \int_t^T \beta_t(u,r) \, dr\right)
            \end{equation}
            holds for $ 0 \leq t \leq T $ and
            $ u \in - \ii [0,1]^n \times \RR^n$.
          \end{itemize}
	\end{theorem}

	\noindent It might be useful for the reader having in mind the
        following expression for the forward characteristics of the
        $(n+d)$-semimartingale $(X,M)$: for all $t$ such that
        $0\leq t\leq T$, we have
	\begin{equation*}
          \exp\left(\eta_t^{(X,M)}(u,v;\, T)\right) = \Econd{\exp\left(\ii \scal{u}{(X_T-X_t)} + \ii \scal{v}{M_T-M_t} \right)}{\cF_t}
	\end{equation*}
	from which it is easier to derive the expression for
        $\kappa_t^{(X,M)}$ for $0\leq t\leq T$.

	\begin{theorem}\label{thm:inversececImpliesDC}
          Let $(\Om, \cF, (\cF_t)_{t\geq 0}, \QQ)$ be a filtered
          probability space together with a tuple $(X, \eta)$ of a
          $n$-dimensional se\-mi\-mar\-tin\-gale $X$ and a
          $\G_n$-se\-mi\-mar\-tin\-gale $\eta$. Furthermore, assume
          that $\eta$ allows for a regular decomposition
          \eqref{regularDec} with respect to a $d$-dimensional
          semimartingale $M$ such that the predictable characteristics
          of $X$ satisfy \eqref{predictable-char} and such that the
          drift condition \eqref{drift-cond} holds, then the
          \emph{conditional expectation condition} holds true.
	\end{theorem}

	\begin{corollary}\label{cor:localIndip}
          Let $(\Om, \cF, (\cF_t)_{t\geq 0}, \QQ)$ be a filtered
          probability space together with a tuple $(X, \eta)$ where
          $X$ is a $n$-dimensional semimartingale and $\eta$,
          $\G_n$-\allowbreak se\-mi\-mar\-tin\-gale, satisfies the
          \emph{conditional expectation condition}.  Moreover, assume
          that $\eta$ allows for a regular decomposition
          \eqref{regularDec} with respect to a $d$-dimensional
          semimartingale $M$ and that the processes $X$ and $M$ are
          locally independent\footnote{See \cite{Kallsen2015} for a
            rigorous definition.}, i.e.
          \begin{equation}
            \kappa_t^{X,M}(u_1, u_2) = \kappa_t^X(u_1) + \kappa_t^M(u_2)
          \end{equation}
          for $u_1 \in \RR^n$ and $u_2 \in \RR^d$. Then
          \begin{equation*}
            \int_t^T \alpha_t(u,r)\,dr = -\kappa_t^M\left(-\ii \int_t^T \beta_t(u,r)\,dr\right)
          \end{equation*}
          for $0\leq t\leq T$ and $u \in \ii[0,1]\times \RR^n$ and,
          furthermore, the conditional expectation
          Condition~\eqref{cec} rewrites as
          \begin{equation*}
            \Econd{\exp\left(\int_s^t \eta_{r-}(u,r)\,dr\right)}{\cF_s} = \exp\left(\int_s^t\eta_s(u,r)\,dr\right)
          \end{equation*}
          for $0\leq s\leq t$.
	\end{corollary}
	\begin{proof}
          To obtain the new form of the conditional expectation
          condition is enough to use \eqref{predictable-char}.
	\end{proof}

	The two previous theorems basically ratify the equivalence
    between the conditional expectation condition on one hand, and
    the short end and drift conditions on the other. In
    \cite{Kallsen2015}, since they assume $\eta$ being of
    L\'evy-Khintchine type for all times, it is possible to show
    equivalence with the fact of $S = \exp(X)$ and $C_t(T, K)$
    being martingales. This is not given for free in our settings,
    but requires additional assumptions (see Remark
    \ref{remarkXi}). For example, requiring that $S = \exp(X)$ is
    a 1-dimensional martingale is equivalent to the condition
    $\eta_s(-\ii, t) = 0$ for all $0\leq s\leq t$. In this case,
    indeed, we can write
    \begin{equation*}
      \Econd{e^{\ii {u}{(X_t-X_s)} }}{\cF_s} = \exp\left(\int_s^t \eta_s(u,r)\, dr\right)
    \end{equation*}
    and for $u = -\ii$ we have $\Econd{e^{\,X_t-X_s}}{\cF_s}=1$
    for all $0\leq s\leq t$.

	\begin{remark}
          Forward characteristics encode the term structure of
          distributions of increments of the stochastic process $X$,
          i.e.~for $0 \leq t \leq T$, the distributions of $X_T - X_t$
          conditional on the information $\cF_t$ at time $t$.  Notice
          that there is redundant information in processes of forward
          characteristics, which then translates into the drift
          conditions \eqref{drift-cond}.
	\end{remark}

	\subsection{Affine processes}
	In this subsection, we introduce affine processes and give
        some important re\-sults on their forward characteristic
        processes.  Moreover, since we are mainly interested in affine
        stochastic volatility models, we will state some properties
        for their particular case.
	
	Let $D$ be a non empty Borel subset of $\RR^d$ to which we
        associate the set
        $\cU \define \{u \in \CC^d : \sup_{x\in D}\Re \scal{u}{x} <
        \infty\}$.
	\begin{definition}[Affine process]
          An affine process is a time-homogeneous Markov process
          $(X_t,\PP^x)_{t\geq 0,x\in D}$ with state space $D$, whose
          characteristic function is an exponentially affine function
          of the state vector.  This means that its transition kernel
          $p_t$ satisfies the following:
          \begin{itemize}
          \item it is stochastically continuous,
            i.e.~$\lim_{s\to t} p_s(x,\cdot)=p_t(x,\cdot)$ weakly on
            $D$ for every $t \geq 0$ and $x\in D$, and
          \item its Fourier-Laplace transform has exponential affine
            dependence on the initial state. This means that there
            exist functions $\Phi:\cU \times \RR_{\geq 0} \to \CC$ and
            $\psi:\cU \times \RR_{\geq 0} \to \CC^{d}$ with
            \begin{equation}\label{charFuncAffine}
              \EE_{x}\left[e^{\langle u, X_t\rangle} \right] =
              \Phi(u,t) e^{\scal{x}{\psi(u,t)}},
            \end{equation}
            for all $x\in D$, $u \in \cU$ and $t\in \RR_{\geq 0}$.
          \end{itemize}
	\end{definition}
	
	\begin{remark}
          The existence of a filtered probability space
          $(\Om, \cF, (\cF_t)_{t\geq 0})$ is already included by the
          notion of Markov process (see \cite{keller-ressel2011}).
	\end{remark}
	
	\begin{remark}
          The definition we gave is not the original provided by
          Duffie et al. in \cite{Duffie2003} but a slightly more
          general one: the right hand side of \eqref{charFuncAffine}
          is equal to $e^{\phi(u,t) + \langle x, \psi(u,t)\rangle}$ as
          long as we know that $\Phi(u,t)\neq 0$, but this can be
          shown (\cite{keller-ressel2013}) and not postulated (as done
          in \cite{Duffie2003}). From now on, we assume
          $\Phi(u,t) = \exp(\phi(u,t))$.  A-priori we do not even have
          a unique definition of the functions $\psi$ and $\phi$, but
          we can assume the normalisation $\phi(u,0)=0$ and
          $\psi^i(u,0)=u$ for all $u\in \cU$ and all $i=1, \dots, d$,
          which makes the functions unique.
	\end{remark}

	In this subsection we build a generic example for term
        structure models for derivatives' prices. Therefore, we define
        an affine stochastic volatility model:
	\begin{definition}[Affine stochastic volatility model]
          Let us consider a proper convex cone $ C \subset \RR^m$ (the
          \emph{stochastic covariance structures}). An \emph{affine
            stochastic volatility model} is a time-homogenous affine
          (Markov) process $(X,Y)$ taking values in $ \RR^n \times C $
          relative to some filtration $(\cF_t)_{t\geq 0}$ and with
          state space $D = \RR^n \times C $ such that
          \begin{itemize}
          \item it is stochastically continuous, that is,
            $\lim_{s\to t} p_s(x,y,\cdot)=p_t(x,y,\cdot)$ weakly on
            $D$ for every $t \geq 0$ and $(x,y)\in D$, and
          \item its Fourier-Laplace transform has exponential affine
            dependence on the initial state. This means that there
            exist (deterministic) functions
            $\phi:\cU \times \RR_{\geq 0} \to \CC$ and
            $\psi_C:\cU \times \RR_{\geq 0} \to \CC^m$ with
			
			\begin{equation}\label{charFuncAffine2}
              \Econd{e^{\scal{u}{X_t} + \langle  v, Y_t \rangle}}{\cF_s} = e^{\phi(u,v,t-s) + \langle u,X_s \rangle + \langle \psi_C(u,v,t-s),Y_s\rangle},
			\end{equation}
			for all $(x,y)\in D$, $0\leq s \leq t$ and
                        $(u,v) \in \cU$, where
			$$
			\cU := \{ (u,v) \in \CC^{n+m} \, | \;
                        e^{\scal{u}{\cdot} + \scal{v}{\cdot}} \in
                        L^{\infty}(D) \},
			$$
			and the normalisations $\phi(u,v,0) = 0 $ and
            $\psi_C^i(u,v,0) = v$ for all $(u,v)\in \cU$
            and $i=1, \dots, m$.
          \end{itemize}
        \end{definition}

	\begin{remark}
          In line with literature on affine processes there is a
          $ \CC^{n+m} $-valued function $\psi$, whose projection onto
          the $ X $-directions is $u$, as exemplified in
          \eqref{charFuncAffine2}. Whence we only need the projection
          in the $ C $-directions, which we denote by $ \psi_C $. This
          corresponds to a standard assumption if we consider $X$ as a
          price process: if we move $X_s$ by a quantity $x$, then also
          $X_t$ gets shifted by the same amount.
	\end{remark}

	Functions $\phi$ and $\psi_C$ are important because they allow
        the introduction of the so-called \emph{functional
          characteristics} (because of complete characterisation) of
        the affine process $(X,Y)$. We define
	\begin{equation}\label{eq:riccati1}
          F(u,v) \define \pd{\phi}{t}(u,v,t)\Bigr|_{t=0^+}, \qquad R_C(u,v) \define \pd{\psi_C}{t}(u,v,t)\Bigr|_{t=0^+}
	\end{equation}
	for all $(u,v) \in \cU$ and continuous in $(0,0)$ (see
        \cite{keller-ressel2013}). Equations \eqref{eq:riccati1} are
        called \emph{Riccati equations}.
	
	More in general, we can also define the \emph{generalised
          Riccati equations}\footnote{The name comes from the fact
          that they boil down to the well-known Riccati equations when
          $(X,Y)$ is a diffusion process.} and prove the following
        theorem (from \cite{keller-ressel-PhD}):
	\begin{theorem}\label{thm:riccati_odes}
          Suppose that $|\phi(u,w,T)|< \infty$ and
          $\norm{\psi_C(u,w,T)}< \infty$ for some
          $(u,w,T) \in \cU \times \RR_{\geq 0}$ . Then for all
          $t \in [0,T]$ and $v$ with $\Re v \leq \Re w$ the
          derivatives \eqref{eq:riccati1} exist. Moreover, for
          $t \in [0,T)$, $\phi$ and $\psi_C$ satisfy the generalised
          Riccati equations:
          \begin{subequations}
            \begin{alignat}{2}
              &\pd{}{t}\phi(u,v,t) = F(u,\psi_C(u,v,t)), \qquad &&\phi(u,v,0) = 0\\
              &\pd{}{t}\psi_C(u,v,t) = R_C(u,\psi_C(u,v,t)), \qquad
              &&\psi_C(u,v,0) = v.
            \end{alignat}
          \end{subequations}
	\end{theorem}
	
	We can also derive the following proposition:
	\begin{proposition}\label{prop:riccati_odes}
          Let $ (X,Y) $ be a homogenous affine process taking values
          in $ D = \RR^n \times C $, then for $t\leq T$ we have that
          \[
            \phi(u,0,t) = \int_0^t F(u,\psi_C(u,0,s))\, ds
          \]
          and
          \[
            \psi_C(u,0,t) = \int_0^t R_C(u,\psi_C(u,0,s))\, ds,
          \]
          where $ (u,v) \mapsto F(u,v) $ and
          $ (u,v) \mapsto \langle R_C(u,v),y \rangle $ are of
          L\'evy-Khintchine form.
	\end{proposition}
	\begin{proof}
          While the first part automatically comes from the definition
          of generalised Riccati equations, the second can be found in
          \cite{keller-ressel-PhD}.
	\end{proof}
	
	\bigskip
	
	\begin{corollary}
          Let $ (X,Y) $ be a homogeneous affine process taking values
          in $ D = \RR^n \times C $ and assume that the finite moment
          condition $ \Esign{\exp ((1+\epsi)\norm{X_t})} < \infty $
          holds true for some $\epsi > 0$, then for $0\leq t\leq T$
          \[
            \eta_t(-\ii u,T)\define F(u,\psi_C(u,0,T-t)) + \scal
            {R_C(u,\psi_C(u,0,T-t))} {Y_t}
          \]
          defines a $ \Gamma_n $-valued semimartingale and the tuple
          $ (X,\eta) $ satisfies the conditional expectation
          condition.
	\end{corollary}
	\begin{proof}
          The proof follows from the previous proposition and simple
          algebraic operations. For any $0\leq t\leq T$, we have
          \begin{align*}
            \Econd{e^{\scal{u}{X_T}}}{\cF_t} &= e^{\phi(u,0,T-t) + \langle u,X_t \rangle + \langle \psi_C(u,0,T-t),Y_t\rangle}\\
             &=e^{\scal{u}{X_t} + \int_{0}^{T-t}F(u,\psi_C(u,0,r))\,dr + \scal{\int_{0}^{T-t}R_C(u,\psi_C(u,0,r))\,dr}{Y_t}}\\
             &=e^{\scal{u}{X_t} + \int_{0}^{T-t}F(u,\psi_C(u,0,r)) + \scal{R_C(u,\psi_C(u,0,r))}{Y_t}\,dr}\\
             &=e^{\scal{u}{X_t} + \int_{t}^{T}F(u,\psi_C(u,0,r-t)) + \scal{R_C(u,\psi_C(u,0,r-t))}{Y_t}\,dr}\\
             &=e^{\scal{u}{X_t} + \int_{t}^{T}\eta_t(-\ii u, r)\,dr}.
          \end{align*}
	\end{proof}

	In interest rate theory, where affine models proved to be a
        powerful tool, Hull-White extensions is realised by making the
        drift term time dependent and plays the fundamental role of
        allowing the calibration of an initial yield curve to the
        prescribed model.  This will be the topic of the next Section.
	
	We see in the following some applications of such theory.
	\begin{example}\textbf{Deterministic term structure of forward
            characteristics: }
          Deterministic forward term structure models correspond to
          time-dependent L\'evy processes.  More precisely, let
          $ (X,\eta) $ be a tuple satisfying the conditional
          expectation condition and assume that $ \eta $ is a
          deterministic, then $X$ is an additive process and
          $\eta_t(u,T)= \eta_0(u,T)$ is of L\'evy-Khintchine form for
          every $T \geq 0$ (compare, for example, with
          Definition~\ref{def:regDecomposition}).  A particular
          example would be any time-dependent L\'evy model.
	\end{example}
	
	\begin{example}\textbf{Interest rate models: }\label{Example:IR}
          If the process $ X $ is one-dimensional, pure-drift and
          absolutely continuous with respect to Lebesgue measure, then
          we fall in the case treated in
          Corollary~\ref{cor:localIndip} and we have
          \[
            \int _t^ T \alpha_t(u,r) \, dr = - \kappa_t^{M}\left(- \ii
              \int_t^T \beta_t(u,r) \, dr\right),
          \]
          but also
          \begin{equation}\label{example:interestRate1}
            \Econd{\exp\left(- \int_t^T \eta_{s-}(u,s)\,ds\right)}{\cF_t} = \exp\left(- \int_t^T \eta_t(u,r) \, dr \right),
          \end{equation}
          from which we obtain
		$$ 
		u X_t = u X_0 - \int_0^t \eta_{s-}(u,s)\, ds.
		$$
		Equation~\eqref{example:interestRate1} is also
                well-known in interest rate theory: if we denote with
                $P(t,T)$ the price of a risk-less zero coupon bond,
                with $f(t,T)$ the forward rate yield prevailing at $t$
                for $T$ and with $r(t)$ the short time interest rate
                at $t$, then we have
		\begin{equation*}
                  P(t,T) = \Econd{e^{-\int_t^Tr(s)\,ds}}{\cF_t} = e^{-\int_t^T f(t,S-t)\,dS}
		\end{equation*}
		for $ 0 \leq t \leq T $ and $ u \in \mathbb{R} $.
                Moreover, if we assumed $M$ being a Brownian motion,
                then we would have
                $\kappa_t^M(u) = -\nicefrac{u^2}{2}$ and, thus,
		$$
		\int _t^ T \alpha_t(u,r) \, dr =
                -\dfrac{1}{2}\left(\int_t^T \beta_t(u,r) \, dr
                \right)^2,
		$$
		from which, differentiating both sides with respect to
                $T$, we obtain the well-known HJM drift condition
		$$
		\alpha_t(u,T) = - \beta_t(u,T) \int_t^T \beta_t(u,r)
                \, dr.
		$$
		Notice that $ {(\eta_{s-}(u,s))}_{s \geq 0} $ is
                linear in $ u $, since $ X $ is pure drift.
              \end{example}
	
              \section{Generalised Hull-White extension}
	
              Hull-White extension of Va\v{s}\'{\i}\v{c}ek model was
              performed adding a time-de\-pen\-dent constant drift to
              the equation for the short term interest rate $r$, in
              order to have a perfect match with the current ($t=0$)
              term structure of forward rates and, thus, to enhance
              calibration.
	
              In this case, we will take a more general approach and
              will encode the extension, represented by a L\'evy
              process, in the constant part of the affine process
              (responsible for the state-independent characteristics
              thereof). In other words, the function $F$ will become
              consequently time-inhomogeneous, thus modifying the
              forward characteristics of the process $X$.

	\begin{corollary}\label{cor:time-inhomogeneous}
          Let $ (\widetilde{X},Y) $ be a time-\emph{inhomogenous},
          homogeneous c\`adl\`ag affine process taking values in
          $ \RR^n \times C $ with time-dependent continuous
          $ T \mapsto F_T $, and assume that the finite moment
          condition
          $ \Esign{\exp ((1+\epsi)\|{\widetilde{X}_t}\|)} < \infty $
          holds true for some $\epsi > 0$, then for $0\leq t\leq T$
          \[
            \widetilde{\eta}_t(-\ii u,T)\define F_T(u,\psi_C(u,0,T-t))
            + \scal {R_C(u,\psi_C(u,0,T-t))} {Y_t}
          \]
          defines a $ \G_n $-valued semimartingale and the tuple
          $ (\widetilde{X},\eta) $ satisfies the conditional
          expectation condition.
	\end{corollary}
	
	\begin{remark}
          Here time-inhomogenous, homogenous affine processes appear
          as generalisation of the approaches in \cite{Carmona2012}
          and \cite{Kallsen2015} (CNKK-approach), since we can
          calibrate a large variety of (virtually, any) initial term
          structure into $t \mapsto F_t$.
	\end{remark}
	
	\begin{remark}
          Although we are only modifying the forward characteristic
          process of $X$, the process $Y$, which is Markov in its own
          filtration, remains the same. This keeps the transformation
          simple and the processes tractable, since it does not affect
          the stochastic covariance structure.
	\end{remark}
	
	The above structure increases the calibration properties of
        the original model. In addition, since the L\'evy process is
        allowed to change over time, we could calibrate it to match
        market conditions for other instant of times (apart the
        initial time).
	
	The main consequence of having such a generalised Hull-White
        extension is that we could compensate fluctuations
        (i.e.~\emph{recalibrations}) in the original model's
        parameters with a calibration of the L\'evy process, so to
        keep the price/volatility surface unchanged. In other words,
        we could consider the parameters of the original model as
        \emph{state variables}. When this is possible, we will talk
        about a model that satisfies the \textbf{\emph{consistent
            recalibration property}}.
	
	A valid question, at this point, would be to know when this is
        possible. Are there conditions that we could impose or verify
        to make sure that such a compensating mechanism can always
        happen?
	
	Let us denote with $(\nu_t^L)_{t\geq 0}$ the L\'evy measure of
    the time-dependent L\'evy process $L$, with $p_t$ and $Z_t$
    the set of parameters and state variables belonging to the
    time-homogeneous model at time\footnote{Since parameters can
      be considered as state variables, they are allowed to change
      in time.} $t$, respectively, and with $\nu^{p_t,Z_s}$ the
    L\'evy measure that has the same expressive capability as the
    original model\footnote{Here, we mean that the price or
      volatility surface created by the model and the L\'evy
      measure should be the same.}, where we made explicit the
    dependence on the parameters $p_t$ and state variables $Z_s$,
    for $s\leq t$.
	
	\begin{proposition}
          Let us assume that the stochastic parameter process
          $(p_t)_{t\geq0}$ has trajectories, whose total variation is
          bounded by a deterministic constant, take values on the
          compact set $\Theta$ of admissible parameters. Moreover,
          assume that $p \mapsto \nu^{p, \cdot}$ is continuously
          differentiable and that $p$ remains constant whenever $Z$
          leaves a prespecified compact set $K$. Then, if for all
          $t \geq 0$ we have the non-negativity condition
          \begin{equation}\label{cond:levyMeasure}
            \nu_t^L \geq \sum_{0\leq s\leq t} \nu_t^{p_s, Z_{s^-}} - \nu_t^{p_{s^-}, Z_{s^-}},
          \end{equation}
          the consistent recalibration property holds.
	\end{proposition}
	\begin{proof}
          The proof is done by induction on the jumping times of the
          parameter process $p$. For more details, see
          \cite{Richter2017}.
	\end{proof}
	
	\bigskip As already mentioned above, this add-on will
        transform the functional characteristic $F$ in a
        time-dependent function and it might be worth noticing how
        this happens in practice.
	
	Using the same notation introduced in \cite{Richter2017}, we
        can define $F_T$ as it appears in Corollary
        \ref{cor:time-inhomogeneous} adding a new time-dependent
        function $\mu$:
	\begin{definition}[$\Inc^D$]
          Let $Z$ be a generic stochastic process with values in the
          (state) space $D$, such that all increments $\Delta Z_s$
          satisfy $z + \Delta Z_s \in D$ for any $z\in D$ and any
          $s\geq 0$. We denote by $\emph{\Inc}^D$ the set which
          contains all continuous functions
          $\mu: \cU\times \RR_{\geq 0} \to \RR$ of the type
          $\mu(u,t) \define \log\EE[\exp(\scal{u}{\Delta Z_t})]$ for
          which $\mu(0,t) =0$ for all $t\geq0$.
	\end{definition}
	In other words, we are adding to the ``old'' $F$ the cumulant
        generating function of the process $(\Delta Z_s)_{s\geq0}$,
        that is $F_s(u,v) \define F(u,v) + \mu(u,v,t-s)$ for all
        $u\in\cU$ and $t\geq s \geq 0$. This will become even clearer in the
        following, when we will specify our consistent recalibration
        model.
	
	Analogously to what already done, we can define $\phitilde$
        and $\psitilde$ as the time-in\-homo\-geneous versions of
        $\phi$ and $\psi$ as solutions to the time-inhomogeneous
        version of the Riccati equations. In particular, for
        stochastic volatility affine processes, similarly to Theorem
        \ref{thm:riccati_odes}, for $s \leq t$ we can write (compare
        with \cite{Richter2017}):
	\begin{subequations}
          \begin{alignat}{2}
            &\pd{}{t}\phitilde(u,v;s,t) = F_t(u,\psitilde_C(u,v;s,t)), \qquad &&\phitilde(u,v;0,0) = 0\\
            &\pd{}{t}\psitilde_C(u,v;s,t) =
            R_C(u,\psitilde_C(u,v;s,t)), \qquad &&\psitilde_C(u,v;0,0)
            = v,
          \end{alignat}
	\end{subequations}
	with $\psitilde_C(u,v;s,t) = \psi_C(u,v;t-s)$. At this point,
        it is also possible to rewrite the expression for the
        \emph{forward characteristics} of the time-inhomogeneous
        process $\widetilde{X}$ as
	\begin{equation}
          \int_t^T \widetilde{\eta}_t(u,r)\, dr = \phitilde(\ii u,0;t,T) + \scal{\psitilde_C(\ii u,0;t,T)}{Y_t}.
	\end{equation}
	In particular, for $t=0$ we recover the characteristic
        function and we obtain
	\begin{equation*}
          \int_0^T \widetilde{\eta}_0(u,r)\, dr = \phitilde(\ii u,0;0,T) + \scal{\psitilde_C(\ii u,0;0,T)}{y}
	\end{equation*}
	and, if we denote with $C(y)$ the set of characteristic
        functions $\widetilde{\eta}_0$, we notice that for any element
        $\widetilde{\eta}_0 \in C(y)$, there exists (at least) one
        $\mu \in \Inc^D$ that defines $\widetilde{\eta}$ itself. It is
        thus possible to establish a surjective function $g$ between
        $\Inc^D$ and $C(y)$. The existence of such $g$ is equivalent
        to nothing but the Condition~\eqref{cond:levyMeasure}
        previously stated since we can consider any jump-time as an
        initial starting point for the process $(\widetilde{X},Y)$ due
        to the Markovianity of the process.
	
	\section{From Heston to Hull-White extended Bates model}\label{sec:HestonBates}
	Before introducing the consistent recalibration model more
        mathematically, let us briefly recall the Heston model, which
        is an affine stochastic volatility model, for $X = \log(S)$
        being the log-return of the underlying price
	\begin{gather}
          \begin{aligned}
            &dX(t)= \Bigl(r-q-\frac{1}{2}V(t)\Bigr) dt + \sqrt{V(t)} dW_1(t), \;\; X(0) = x_0,\\
            &dV(t)=k \left[\theta -V(t)\right]dt + \sigma \sqrt{V(t)} dW_2(t), \;\; V(0) = v_0,\\
            & dW_1(t)\, dW_2(t) = \rho\, dt, \qquad \rho \in [-1,1],
          \end{aligned}
	\end{gather}
	and where $r$ and $q$ represent the instantaneous risk-free and dividend yields respectively and are constant, $\theta > 0$ is the long-term mean of the variance, $k > 0$ is the speed of mean-reversion, $\sigma > 0$ represents the instantaneous volatility of the variance process $V$. In order to ensure positivity of the variance process, we need to satisfy $2 k \theta > \sigma^2$ (Feller condition).\\
	For $0\leq t\leq T$, we have that $\eta$ defines a
        $\G_1$-semimartingale:
	\begin{equation*}
          \eta_t(u, T) = F(\ii u, \psi_C(\ii u, 0, T-t)) + R_C(\ii u, \psi_C( \ii u, 0,T-t))\, V_t,
	\end{equation*}
	where $C$ coincides with $\RR_{>0}$ and
	\begin{subequations}
          \begin{alignat*}{2}
            &F(u_1,u_2) &&= k\theta u_2 + (r-q)u_1,\\
            &R_C(u_1,u_2) &&= \frac{1}{2}u_1(u_1-1) +
            \frac{1}{2}\sigma^2u_2^2 + \sigma\rho u_1 u_2 - k u_2.
          \end{alignat*}
	\end{subequations}
	
	The Hull-White extension of the Heston model consists in a
    generalised version of the so-called Bates model in which we
    add a compensated\footnote{Compensation is necessary to have a
    martingale process, as it is often the case in the pricing
    context.} jump L\'{e}vy process $L$ with L\'{e}vy measure
    $\nu(t,dx)$ to the dynamics of the log-return $X$:
	\begin{equation*}
          dX(t)= \Bigl(r-q-\frac{1}{2}V(t)\Bigr) dt + \sqrt{V(t)} dW_1(t) + dL_t.
	\end{equation*}
	The first consequence that should appear obvious is that we
        are enriching the space of calibrated volatility surfaces,
        thanks to the L\'{e}vy process, while keeping the same
        dimensions of the state variables. Accordingly, the functional
        characteristic $F$ will then change to
	\begin{equation}
          F_t(u_1,u_2) = k\theta u_2 + (r-q)u_1 + \mu_L(u_1,u_2,t),
	\end{equation}
	where $\mu_L$ is the cumulant generating function of $L$.  As
        we will see below, we can establish a bijective relation
        between $\mu_L$ and the L\'evy measure $\nu_L$.
	
	We are talking about a generalised Bates model since the
        L\'{e}vy measure is also allowed to change in time and, as
        already said, this permits to make also other parameters time
        dependent.

	\subsection{Consistent recalibration (with words)}\label{subsec:CRCwords}
	Time is mature to explain how the generalised Bates model can
        be used as a consistent recalibration (CRC) model.  Recall
        that although formally there are only two state variables, now
        also parameters are free to change in time thanks to the
        compensation mechanism that the Hull-White extension provides.
	
	Let us start at time $t=t_0$ with a log-price $X_{t_0}$, a set
        of parameters
        $p_{t_0} = (r, q, k, \theta_{t_0}, \sigma_{t_0}, \rho_{t_0})$,
        an initial variance for the log-price $V_{t_0}$ and the
        compensated jump L\'{e}vy process $L_{t_0}$ which represents
        the Hull-White extension.  Note that some of the parameters
        are not constant, but change over time and are denoted by the
        time-index.  This particular combination of state variables
        and parameters fully specifies a particular model $\cM_1$
        among all possible models $\cM_i$ that can represent an
        implied volatility surface (IVS) without breaking any
        no-arbitrage constraints and should be able to reflect those
        market conditions that are summarised by the IVS at time
        $t_0$. It is thus natural to write
        $\text{IVS}_t = \text{IVS}(X_t,V_t; \theta_{t}, \sigma_{t},
        \rho_{t}; \{T_i\}, \{K_j\})$ for the volatility surface at
        time $t$. The model state variables $X$ and $V$ are thus able
        to evolve in time until $\cM_1$ is able to mirror the
        market. Eventually, this situation will break at time $t=t_1$
        and a new calibration will be necessary.
	\begin{enumerate}
        \item Starting from time $t_0$, $X$ and $V$ can evolve until
          $t_1 = t_0 + \Delta t$, where the new volatility surface is
          given by
          $\text{IVS}(X_{t_1},V_{t_1}; \theta_{t_0}, \sigma_{t_0},
          \rho_{t_0}; \{T_i\}, \{K_j\})$.
        \item Parameters $(\theta_{t_0}, \sigma_{t_0}, \rho_{t_0})$
          will move to another configuration
          $(\theta_{t_1}, \sigma_{t_1}, \rho_{t_1})$, but, to
          enforce a smooth change between the first configuration and
          second,
        \item Also the L\'{e}vy process will be adjusted and will
          compensate the changes in the parameters
          $(\theta_{t}, \sigma_{t}, \rho_{t})$ to reproduce the same
          IVS.
        \item In this way we can represent realistically the behaviour
          of the market.
	\end{enumerate}
	The \emph{recalibration} of the Hull-White extension is also
        preserving the drift-con\-di\-tion of the forward
        characteristics, thus ensuring that we do not violate
        no-arbitrage constraints.  The new model $\cM_2$ will then
        specify the evolution in time of the state variables until
        another recalibration will be needed.
	
	\begin{figure}
          \centering \scalebox{0.9}{%
            \begin{tikzpicture}
              \coordinate (p0) at (-2.2,-1.4); \coordinate (c) at
              (0,0); \coordinate (p1) at (-0.2,0.5); \coordinate (p2)
              at (0.8,-0.6); \coordinate (p3) at (-1.8,1.2);
              \coordinate (p4) at (0.5,0.3); \coordinate (p5) at
              (2.8,1.3);
			
              \draw (c) ellipse (4cm and 3cm);
              \path[postaction={decorate, decoration={ raise=0.3em, text along path, text={|\Large\bf|IVS models}, text align=center, }, }]
              (c) ++(220:4cm and 3cm) arc(220:340:4cm and 3cm);
			
              \fill[black] (p0) circle (0.03cm) node[below] {$t=t_0$};
              \fill[black] (p1) circle (0.03cm) node[above] {$t=t_1$};
              \fill[black] (p5) circle (0.03cm) node[above] {$t=t_2$};
			
              \draw[->] (p0) -- (p1) node[pos=.5,sloped,above]
              {$\cM_1$}; \draw[->] (p0) -- (p2)
              node[pos=.5,sloped,below] {$\cM_2$}; \draw[->] (p0) --
              (p3) node[pos=.5,sloped,above] {$\cM_i$};
			
              \draw
              [->,decorate,decoration={snake,amplitude=.4mm,segment
                length=2mm,post length=1mm}] (p1) -- (p4); \draw
              [red!50!black, -Stealth] (p4) .. controls (0.,0.)
              .. (p1); \draw[->] (p1) -- (p5)
              node[pos=.5,sloped,above] {$\cM_2$};
            \end{tikzpicture}
          }
          \label{fig:M1}
          \caption{At $t=t_1$ the change in the parameters
            $(\theta_{t_0}, \sigma_{t_0}, \rho_{t_0})$ would cause a
            jump in the volatility structure (snake arrow), but this
            is compensated by the change in the L\'{e}vy process $L$
            (red bent arrow).}
	\end{figure}
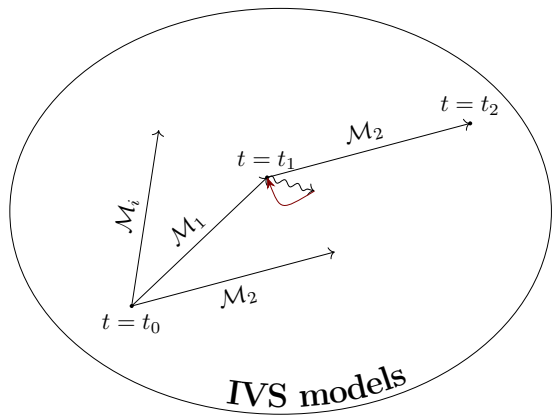
	
	Notice that the evolution of the state variables is determined
        by the ``old'' parameters
        $(\theta_{t_i}, \sigma_{t_i}, \rho_{t_i})$ on the closed
        interval $[t_i, t_{i+1}]$ and for this reason there is no
        discontinuity at $t=t_{i+1}$ in the modelled IVS caused by the
        parameters movement.
        Further, since we do not know in advance the exact 
        recalibration times $t_i$, these are random variables (more in 
        Section \ref{sec:CRC_math}).

	\section{CNKK Equation}
	\subsection{Quick heuristics}\label{subsec:heuristics}
	The mathematical formulation that is needed to describe what
        we sketched above starts from \eqref{regularDec}.  If we
        rewrite the same equation using Musiela's parametrisation,
        defining $x := T-t$, then the map becomes
        $(u, t, x) \mapsto \eta_t(u, t+x)$ in the new notation.  In
        addition, let us introduce the strongly continuous semi\-group
        $\{S(t) \,|\, t \geq 0\}$ of right shifts, such that for a
        proper function $g$ this is mapped to
        $S(t)g(u, \cdot) = g(u, t+\cdot)$.  Then we can rewrite
        Equation~\eqref{regularDec} as
	\begin{equation}
          \eta_t(u,t+x) = S(t)\eta_0(u,x) + \int_0^t\!\! S(t-s)\alpha_s(u,t+x)  ds + \sum_{i=1}^d \!\int_0^t\!\! S(t-s)\beta^i_s(u,t+x)  dM^i_s,
	\end{equation}
	which can be rewritten in terms of $\theta_t(u,x) := \eta_t(u, t + x)$ and, with abuse of notation,
        $\alpha_t(u,x) := \alpha_t(u,t+x)$,
        $\beta_t(u,x) := \beta_t(u,t+x)$ as
	\begin{equation}\label{eq:pre-CNKK}
          \theta_t(u,x) = S(t)\theta_0(u,x) + \int_0^t\! S(t-s)\alpha_s(u,x)  ds + \sum_{i=1}^d \!\int_0^t\! S(t-s)\beta^i_s(u,x)  dM^i_s.
	\end{equation}
	Finally, the passage to the limit will justify what written in the next subsection.

	\subsection{CNKK SPDE}
	The framework we will develop in the following will allow a
        thorough analyses of factor models in the CNKK-approach
        introduced in \cite{Carmona2012} and \cite{Kallsen2015}, yet
        with crucial differences. For example, as already done by
        Kallsen and Kr\"{u}hner, we assume that the volatility
        processes $\beta^i$ of the forward characteristic $\eta$ are
        functions of the present state of $\eta$ itself, i.e. for all
        $i=1, \dots, d$
	\begin{equation*}
          \beta_t^i(u, T)(\omega) = \sigma\big(t,\eta_{t^-}(\cdot,\cdot)(\omega)\big)(u,T),
	\end{equation*}
	but, since we introduce the right-shift operator, we will
        obtain an SPDE and not simply an SDE.
	
	\begin{definition}[L\'evy codebook Hilbert
          space]\label{def:HilbertCodebook}
          Let $ G $ be a Hilbert space of continuous complex-valued functions defined on the strip $ - \ii [0,1]^n \times \mathbb{R}^n $, i.e. $ G \subset C((- \ii [0,1]^n) \times \mathbb{R}^n; \mathbb{C}) $.\\
          $ H $ is called a \emph{L\'evy codebook Hilbert space} if
          $H$ is a Hilbert space of continuous functions
          $ \eta: \mathbb{R}_{\geq 0} \to G $,
          i.e.~$H \subset C(\mathbb{R}_{\geq 0};G)$ such that
          \begin{itemize}
          \item There is a continuous embedding
            $ H \subset C(\mathbb{R}_{\geq 0} \times (- \ii [0,1]^n)
            \times \mathbb{R}^n; \mathbb{C})$,
          \item The shift semigroup
            $ (S_t \eta) (u,x) := \eta(u,t+x) $ acts as strongly
            continuous semigroup of linear operators on $ H $,
          \item Continuous functions of finite activity
            L\'evy-Khintchine type
            \[
              (u, t) \mapsto \ii \scal{a(t)}{u} -
              \frac{\scal{u}{b(t)u}}{2} + \int_{\mathbb{R}^n}
              \big(\exp(\ii \scal{\xi}{u}) -1\big) \,\nu_t(d \xi)
            \]
            lie in $ H $, where $ a $, $ b $, $ \nu $ are continuous
            functions defined on $ \mathbb{R}_{\geq 0} $ taking values
            in $ \mathbb{R}^n $, the positive-semidefinite matrices on
            $ \mathbb{R}^n $ and the finite positive measures on
            $\RR^n$. This is for example the case for processes with
            independent increments and finite
            variation. 
          \end{itemize}
	\end{definition}
	
	\begin{remark}
          Notice that we do not assume that there are additional
          stochastic factors outside the considered parametrisation of
          liquid market prices.
	\end{remark}
	
	\begin{remark}
          Notice that elements of the Hilbert space $H$ are understood
          in Musiela parametrisation and therefore denoted by a
          different letter in the sequel. As already written in
          Subsection~\ref{subsec:heuristics}, we have the relationship
          $ \eta_t(u, t+x) = \theta_t(u,x) $, with $ x := T-t $.  In
          this sense, we also have the equality
          $ \theta_t(u, 0) = \kappa_t^X(u) $ for the predictable
          characteristics of $X$.
	\end{remark}

	\begin{definition}[CNKK equation]
          Let $H$ be a L\'evy codebook Hilbert space. We call the
          following stochastic partial differential equation
          \begin{equation}\label{eq:CNKK}
            d \theta_t = \big(A \theta_t + \mu_{\operatorname{CNKK}}(\theta_t)\big) dt + \sum_{i=1}^d \sigma_i(\theta_t) \, dB^i_t
          \end{equation}
          a \emph{CNKK equation} $ (\theta_0,\kappa,\sigma) $ with
          initial term structure $ \theta_0 $ and characteristics
          $ \kappa $ and $ \sigma $, if
          \begin{itemize}
          \item $ A = \frac{d}{dx} $ is the generator of the shift
            semigroup on $H$,
          \item $ \sigma_i: U \subset H \to H$, $ U $ an open subset
            of $H$, are locally Lipschitz vector fields, and
          \item $ \mu_{\operatorname{CNKK}} : U \to H $ is locally
            Lipschitz and satisfies that for all $ \eta \in \Gamma_n $
            we have
            \begin{equation}\label{eq:CNKK_drift}
              \int _0^ {T-t}\!\!\! \mu_{\operatorname{CNKK}}(\theta)(u,r) \, dr = \theta(u,0) - \kappa_{\theta}\!\left(u, - \ii \int_0^{T-t}\!\!\! \sigma(\theta)(r,u) \, dr \,; 0\!\right)\!,
            \end{equation}
            where $ {(\kappa_{\theta})}_{\theta \in U} $ is
            $ \Gamma_{n+d} $-valued for each $ \theta \in \Gamma_n $,
            such that $ \kappa_{\theta}(u,0;0) = \theta(u,0) $ and
            $ \kappa_{\theta}(0,v;0) = -\frac{\norm{v}^2}{2} $, for
            $ u \in \RR^n $, $ v \in \RR^d $.
          \end{itemize}
	\end{definition}

	\begin{remark}
          $\kappa_{\theta}$ is the forward characteristic process
          associated to the couple $(X, B)$, where $X$ is (still) the
          log-return price process and $B$ the driving process of
          $\theta$.  Moreover, Equation \eqref{eq:CNKK_drift} can be
          seen as a \emph{drift-condition} and it is analogous to
          Equation \eqref{drift-cond}, reformulated under the
          Musiela's parametrisation.
	\end{remark}
	
	\begin{remark}
          It is evident how Equations \eqref{eq:CNKK} and
          \eqref{eq:pre-CNKK} relate to each other and how the former
          can be seen as the limit case of the latter.
	\end{remark}
	
	\begin{remark}
          We do not require that all solutions of Equation
          \eqref{eq:CNKK} are $ \Gamma_n$-valued, which would be too
          strong as a condition and difficult to characterise. In
          particular, $\Gamma_n$ is a more general than what we need.
	\end{remark}

	\begin{proposition}
          Let $ \theta $ be a $ \Gamma_n $-valued solution of a CNKK
          equation and let $ X $ be a semimartingale such that the
          predictable characteristics satisfy
          \begin{equation*}
            \kappa_t^{(X,B)}(u,v) = \kappa_{\theta_t}(u,v;t)
          \end{equation*}
          for $ u \in \RR^n $, $ v \in \RR^d $ and $ t \geq 0 $, then
          the tuple $ (X,\theta) $ satisfies the conditional
          expectation condition.
	\end{proposition}
	\begin{proof}
          The proposition is basically a consequence of
          Theorem~\ref{thm:inversececImpliesDC}: the drift condition
          is satisfied by assumption and
          \begin{equation*}
            \exp \left(\ii \scal{u}{X_t} - \int_0^t \kappa_{s-}^{(X,B)}(u,v;s)\,ds \right)
          \end{equation*}
          is a (local) martingale because of the L\'evy-Khintchine
          assumption in Definition~\ref{def:HilbertCodebook} regarding
          the functions of the L\'evy codebook Hilbert space.
	\end{proof}

	\subsection{Generalisation of the HJM equation}
	
	Equation~\eqref{eq:CNKK} is very similar to the famous HJM
        equation\footnote{In the literature, this is also known as
          HJMM equation, where the last M stands for Musiela.}, but
        there are relevant differences, for example both the drift and
        drift condition are different and, what is more, functions
        have another argument (a \emph{strike} dimension) that is
        completely missing in the case of the HJM equation, where only
        a time-dimension is considered.
	
	We can construct a particular example which corresponds indeed
        to the HJM equation: let us consider a situation without
        leverage (where the Brownian motion $B$ is independent of the
        return process $X$), assuming that
	\begin{equation*}
          \kappa_\theta(u,v;0) = \theta(u,0) - \frac{\norm{v}^2}{2},
	\end{equation*}
	for $u \in \RR^n$, $v\in \RR^d$ and $t\geq 0$. Basically, we
        are looking at functions in a restricted L\'evy space, for
        which the L\'evy measure is null.  This implies that the CNKK
        equation is a parameter-dependent HJM equation. In this case,
        Condition~\eqref{eq:CNKK_drift} can be simplified to
	\begin{equation*}
          \mu_{\operatorname{CNKK}}(\theta)(u,x) = -\sum_{i=1}^d \sigma^i(\theta)(u,x) \int_0^x \sigma^i(\theta)(u,s) \,ds,
	\end{equation*}
	for $ x \geq 0 $ and $ u \in \RR^n $ (note the analogies with
        Example~\ref{Example:IR}).
	
	\begin{example}\textbf{Black-Scholes model: }
          It might be interesting at this point to see a concrete
          example coming from a simpler model.  If we consider asset
          prices described by a geometric Brownian motion
          $dS_t = S_t \sigma dW_t$, where $\sigma > 0$ is constant and
          $W$ is a standard Brownian motion, then the log-prices $X$
          are given by
          $dX_t = d\log S_t = \nicefrac{\sigma^2}{2}\,dt + \sigma
          dW_t$.  We find that
          $\eta_t(-\ii u, r) = \nicefrac{1}{2}\,\sigma^2 u(u-1) =
          F(u)$, in the notation of \eqref{eq:riccati1}. It is easy to
          see that $\eta$ is pure-drift and that the extended
          functional characteristic becomes
          $F_t = \frac{1}{2}\sigma_t^2 u(u-1) + \mu_L(u,t)$, where
          $\mu_L$ is the cumulant of the generalised Hull-White
          extension.
	\end{example}

	From what has been said so far, it is clear how the CNKK
        equation is in fact a generalisation of the HJM equation.
	\begin{remark}
          It is possible to further increase the complexity of the
          equation, for example considering options on a term
          structure. In this case, we would need another argument to
          take into account for both drift and volatility.
	\end{remark}
	
	\begin{remark}
          All these considerations are conceivable only because we are
          dealing with affine processes. In general, for a return
          price process $X$, it might not be possible to write the
          conditional expectation condition and to continue with the
          following statements.
	\end{remark}
	
	\begin{remark}
          It is possible to generalise what has been said in this
          section for processes driven by infinite dimensional
          Brownian motions, e.g.~in Equation~\eqref{eq:CNKK} we could
          replace the sum $\sum_{i=1}^d$ with $\sum_{i \in \NN}$.  The
          theory has been paved in the book by Da Prato and Zabczyk
          \cite{da_prato_zabczyk_2014}, but also Chapter 2 of
          \cite{filipovic2001consistency} by Filipovi\'c provides a
          useful and accessible introduction.  We continue considering
          the finite-dimensional case because this does not really
          create new hurdles to be solved, in contrast to the main
          obstacle, the drift $\mu_{\operatorname{CNKK}}$, for which
          no explicit expression is available.
	\end{remark}

	\section{Consistent recalibration (with maths)}\label{sec:CRC_math}
	We are now ready to face the same considerations we reported
        above in more rigorous settings. First of all, let us recap
        the most important equations. For the return process, we have
	\begin{gather}\label{eq:returnProcess}
          \begin{aligned}
            &dX(t)= \delta_t^X(X_t, V_t) \, dt + \gamma_t^X(X_t, V_t) \, dW_1(t) + dL_t, \;\; X(0) = x_0,\\
            &dV(t)= \delta_t^V(X_t, V_t) \, dt + \gamma_t^V(X_t, V_t)
            \, dW_2(t), \;\; V(0) = v_0,
          \end{aligned}
	\end{gather}
	with $dW_1(t)\, dW_2(t) = \rho_t\, dt$, for $\rho_t \in
        [-1,1]$.  Drifts and volatility coefficients are denoted by
        $\delta$ and $\gamma$ respectively and can be functions of $X$
        and $V$, e.g.~$\gamma^X(x,v) = \gamma^V(x,v) = \sqrt{v}$.
        While for the forward characteristic process, our
        \emph{codebook}, we report the CNKK SPDE
	\begin{equation}\label{eq:thetaCodebook}
          d \theta(t) = \big[A \theta(t) + \mu_{\operatorname{CNKK}}(\theta(t))\big]\, dt + \sum_{i=1}^d \sigma_i(\theta(t)) \, dB^i(t), \;\; \theta(0)=\theta_0,
	\end{equation}
	where the dependence among the Brownian motions $B^i$ with
        $i=1,\dots,d$ and these with $W_j$ for $j=1,2$ is not
        specified.  The key relation that connects the two different
        systems is given in Corollary \ref{cor:time-inhomogeneous} and
        is the following (rewritten in the Musiela notation):
	\begin{equation}\label{eq:keyRelation}
          \theta_t(-\ii u, x) = F_{t+x}(u, \psi_C(u,0;x)) + \scal{R_C(u,\psi_C(u,0;x))}{V_t}.
	\end{equation}
	Last but not least, we should also remember that parameters
        are free to move in time. As such, we consider the process
        $p$, whose dynamics are exogenously given, but which are
        confined inside the space of admissible parameters
        $\Theta \subset \RR^M$. In the example described in
        Subsection~\ref{subsec:CRCwords} it is defined as
	\begin{equation*}
          p_t = (\theta_t, \sigma_t, \rho_t), \qquad t \geq 0,
	\end{equation*}
	given the constraints of positivity and Feller condition for $\sigma_t$ and $\theta_t$ and $\rho_t \in [-1,1]$. This is the reason why we used the subscript $t$ in Equations~\eqref{eq:returnProcess} above.\\
	Let us suppose that at time $t_0$ the model can fit well the
        market surface given by observed call prices (or,
        equivalently, of implied volatility surface)
        $C_{t_0}^{\emph{obs}}(T_i,K_j)$ for $i=1,\dots,n$ and
        $j=1,\dots,m$. In this condition, the process
        $(\theta, (X,V))$ are free to evolve in time until the a new
        calibration is necessary. This is the case when
	\begin{equation}\label{eq:calibrCond}
          \Delta_{C_{t_0}} \define \sum_{i=1}^n \sum_{j=1}^m \left|C_{t_0}^{\operatorname{model}}(T_i, K_j) - C_{t_0}^{\operatorname{obs}}(T_i, K_j)\right|^2 > \epsi,
	\end{equation}
	where $C_{t}^{\emph{model}}$ is the price of a call option at time $t$ given by the model and $\epsi$ is a threshold fixed a priori.\\
	Thus, we can define the following hitting times: for $i \in
        \NN$,
	\begin{gather}
          \begin{aligned}
            \tau_0\,\,\, &\define \inf\left\{t > t_0 \,:\, \Delta_{C_{t}} > \epsi\right\}\\
            \tau_{i+1} &\define \inf\left\{t > \tau_i \,:\,
              \Delta_{C_{t}} > \epsi\right\}.
          \end{aligned}
	\end{gather}
	As already underlined in \cite{Kallsen2015}, the model price
        $C_{t}^{\emph{model}}$ can be expressed as a measurable
        function of the codebook $\theta$ satisfying
        Equation~\eqref{eq:thetaCodebook}. In particular, since it is
        a progressively measurable process (it is right-continuous on
        a complete probability space), we can use D\'ebut theorem,
        which tells us that the sequence $(\tau_i)_{i=\NN_0}$ is in
        fact made by \emph{stopping times}.

	\begin{remark}
          Since we are in the same framework as
          \cite{da_prato_zabczyk_2014}, and indeed we could generalise
          all results to infinite dimensional Brownian motion, it is
          worth mentioning that the solution process $\theta$
          satisfies the strong Markov property.
	\end{remark}
	
	The strictly increasing sequence $(\tau_i)_{i=\NN_0}$ is
        important since it is at these random (stopping) times that we
        have to run a \emph{new calibration} procedure for the model.
        Both the ``true'' state variables $X$ and $V$ and the
        parameter process $p$ are allowed to change in order to have
        $\Delta_{C_{t_0}}$ less than $\epsi$ again.  This is just the
        only first calibration problem we need to solve. Indeed, in
        order to compensate the changes caused by the new parameters,
        we have to modify $F_t$.  In particular, we can ca\-li\-brate
        the so-called Hull-White extension part, which enters $F_t$ as
        the cumulant generating function $\mu_L$ of the L\'evy process
        $L$.  This \emph{re-calibration} ensures that we are not
        breaking the validity of Equation~\eqref{eq:keyRelation},
        while allowing for an ``exact'' match (in the sense of having
        $\Delta_{C_{t}} < \epsi$) with the observed data.  Once
        $\mu_L$ is recovered, we are able to write down again the
        equation for the codebook $\theta$.  We can summarise this
        last passage more mathematically by introducing an operator
        $\cI$ such that
	\begin{equation}\label{eq:operatorI}
          \cI: \Theta \times \RR_+^{n+m}\to \Inc^{\RR^n\times C}, \qquad \left(p,(C^{\emph{obs}})\right) \mapsto \mu_L.
	\end{equation}
	\begin{remark}
          The cumulant generating function $\mu_L$ identifies uniquely
          $L$ if and only if the process $L$ has finite moments of
          order $n$ for all $n \in \NN$. If we denote with $\nu_L$ the
          L\'evy measure associated to $L$, then this is true if and
          only if
          \begin{equation*}
            \forall n \in \NN, \;\; \int_{\|x\|\geq 1} \|x\|^n\,\nu_L(t,dx) < \infty.
          \end{equation*}
	\end{remark}
	
	Eventually, we are now ready to give the definition of
    Consistent Recalibration (CRC) model with piecewise constant
    parameters $p$:

\begin{definition}[Consistent Recalibration Model with     Piecewise-constant $p$]\label{def:crcModel}
  Let $(\Omega, \cF, \FF, \PP)$ be a complete filtered
  probability space. The quintuple $(\theta, (X,V)$, $p$,
  $L, (\tau_i)_{i\in \NN_0})$ is called \emph{consistent
    recalibration model} for equity derivative pricing if for
  the stochastic processes $(\theta, (X,V), p)$ with values in
  $H \times (\RR^n\times C) \times \Theta$ there exists a jump
  L\'evy process $L$ (with finite moments) such that the
  following conditions are satisfied for all $n \in \NN_0$:
  \begin{enumerate}
  \item[(i)] The Hull-White extension $L$ on
    $[\tau_n, \tau_{n+1}]$ is determined by calibration to
    $\theta(\tau_n)$ through $\mu_L$:
    \begin{align*}
      &\theta(\tau_n)(u,0) = \kappa_{\theta(\tau_n)}(u,0;0),\\
      &\mu_{L(\tau_n)} = \cI\left(p(\tau_n),(X(\tau_n),V(\tau_n)),C_{\tau_n}^{\emph{obs}}\right),
    \end{align*}
    and for $t \in [\tau_n, \tau_{n+1}]$ we have
    \begin{equation*}
      L(t) = S(t-\tau_n)L(\tau_n).
    \end{equation*}
	
  \item[(ii)] The evolution of $(X,V)$ on
    $[\tau_n, \tau_{n+1}]$ corresponds to the Hull-White
    extended stochastic volatility affine model determined by
    the parameters $p(\tau_n)$ and by the process $L(\tau_n)$:
    \begin{equation*}
      (X,V)(t) = \left(X^{\tau_n, X(\tau_n)}, V^{\tau_n, V(\tau_n)}\right)(t) \;\mbox{ for }\; t \in [\tau_n, \tau_{n+1}],
    \end{equation*}
    where $(X^{s,x},V^{s,v})$ is the unique solution of the
    system of SDEs~\eqref{eq:returnProcess} on $[s,\infty)$
    with initial conditions $X(s) = x$ and $V(s)=v$ and with
    $L_t$ replaced by $L_{t-s}$. Implicitly, we also assume
    that all parameters $p$ that enters the model are
    admissible (with the usual meaning).
    			
  \item[(iii)] The evolution of $\theta$ on
    $[\tau_n, \tau_{n+1}]$ is determined by $X$ and $V$
    according to the prevailing Hull-White extended stochastic
    volatility model: for $t \in [\tau_n, \tau_{n+1}]$ and
    $x \in [0, \tau_{n+1}-\tau_n]$
    \begin{equation*}
      \theta(t)(-\ii u,x) = F_{\tau_n + x}(u, \psi_C(u,0;x)) + \scal{R_C(u,\psi_C(u,0;x))}{V(t)}.
    \end{equation*}
  \end{enumerate}
\end{definition}
	
For the processes $(X,V)$ and $\theta$, we use the same
symbols as in Equations~\eqref{eq:returnProcess} and
\eqref{eq:thetaCodebook}, with a slight abuse of notation,
since these stochastic processes evolve in the intervals
$[\tau_n, \tau_{n+1}]$ following the same dynamics, but
according to the parameters $p(\tau_n)$ and to the process
$L(\tau_n)$.
The parameters $p$ remain constant in each interval of the type $[\tau_n, \tau_{n+1})$ and, by construction, $(\theta, (X,V))$ is continuous on every stopping time $\tau_n$.\\
In this sense, any CRC model can be seen as the concatenation
of stochastic volatility affine models with static parameters.

\bigskip
\noindent We can now prove the following.
\begin{theorem}
   Let $(\theta, (X,V)$, $p$, $L, (\tau_i)_{i\in \NN_0})$ be a consistent
   recalibration model as in Definition \ref{def:crcModel} and $S=\exp(X)$ 
   the discounted price process. Then $S$ and European call 
   (resp. put) option prices $C_t(T, K)$ (resp. $P_t(T, K)$) 
   are (true) martingales on $\RR_{\geq0}$. 
   Moreover, if we denote the payoff function of a call option 
   as $V(y) = (y-K)_+$, then the following pricing formula holds:
   \begin{equation}\label{eq:callPrice}
       C_t(T,K) = \frac{e^{rt}}{2\pi}
       \int_{\ii\Im(\Tilde{z})+\infty}^{\ii\Im(\Tilde{z})-\infty}
       \Psi_T(-z) \hat{V}(z)\,dz,
   \end{equation}
   where $\tau_{n+1}> T \geq \tau_n$ for some $n \in \NN$,
   $\Tilde{z}$ belongs to the analytic strip for which we have
   finite exponential moment (cf. Remark \ref{remark:strip}), 
   $\hat{V}$ denote the Fourier transform of $V$ and $\Psi_{T|\tau_n}$ is 
   the characteristic function of $X_T$, i.e.
   \begin{equation}\label{eq:Psi}
       \Psi_{T|\tau_n}(u) = \EE\left[e^{\ii \langle u, X_T\rangle}\,|\,\cF_{\tau_n}\right] = e^{\ii\langle u, X_{\tau_n}\rangle + \int_0^T \theta_{\tau_n}(u, r)\,dr}.
   \end{equation}
\end{theorem}
\begin{proof}
    From point \emph{(ii)} of Definition \ref{def:crcModel} we have 
    the couple $(X,V)$ is a stochastic volatility model on each interval
    of the type $[\tau_n, \tau_{n+1}]$, from which martingality of $S$ follows 
    by taking the conditional expectation and noting that, by Definition
    \ref{def:HilbertCodebook}, we automatically have $\theta(0,x) = 0$ 
    for any $x \geq 0$. We can further extend this result
    on $\RR_{\geq 0}$ because, by our own construction, the
    concatenation of the entire process is made in a continuous way.\\
    Once it is established that $S$ is a martingale, the same
    property follows for European call and put options by use of
    the tower property of the conditional expectation.\\
    The pricing formula \eqref{eq:callPrice} comes from Fourier pricing,
    which was initially introduced by Carr and Madan in \cite{CarrMadan},
    while \eqref{eq:Psi} comes from Corollary \ref{cor:time-inhomogeneous}
    where the characteristic process has been expressed in Musiela
    notation.
\end{proof}

\begin{remark}
      Having shown that the discounted price process $S$ is a martingale,
      we automatically rule out arbitrage possibilities for the so-called
      consistent recalibration models introduced in 
      Definition \ref{def:crcModel}.
\end{remark}

\begin{remark}
      Note that the from \eqref{eq:Psi} we see that the characteristic
      process $\theta$ encodes information on the conditional expectation
      of $X$, which is equivalent, as already noted in \cite{Richter2017}
      to the knowledge of the entire derivative-price surface, thanks
      to Breeden–Litzenberger formulas.
\end{remark}

\begin{remark}
      Equation \eqref{eq:callPrice} can be generalised to other
      payoff function and is usually enriched with a dampening factor
      which is introduced to exploit numerical algorithm (\cite{CarrMadan}).
      In our case, we will use Fourier-pricing techniques in order to
      obtain the dataset for a supervised learning algorithm.
\end{remark}

\subsection{Numerical considerations}
There are practical remarks that we should consider when
    dealing with CRC models as defined above in a numerical
    framework.

\begin{itemize}
    \item \textbf{\emph{Simulations}}: As already mentioned in
      \cite{paperYieldCurve2018}, it is worth noting that
      simulating $(X,V)$ is much easier than simulating the HJM
      codebook, that is $\theta$, in particular when this is
      infinite dimensional, as it could be the case also here. In
      fact, with the approach we are outlining, we do not need to
      simulate anything from any infinite dimensional
      distribution.
	
    \item \textbf{\emph{Drift term}}: Even if we wanted to
      simulate $\theta$ solving the SPDE~\eqref{eq:thetaCodebook},
      we should be able to write down explicitly the drift term
      $\mu_{\operatorname{CNKK}}(\theta)$, but, apart from some
      degenerate cases, this is not possible. The only way to
      overcome this chasm is acknowledging
      Equation~\eqref{eq:keyRelation} as a key relation for the
      entire construction.
	
    \item \textbf{\emph{Process}} \boldsymbol{$p$}: If we assume
      that a piecewise process for the parameters $p$ is given (or
      obtained through calibration), then CRC models can be
      simulated following steps $(i)$ to $(iii)$ in
      Definition~\ref{def:crcModel}.
	
    \item \textbf{\emph{Operator}} \boldsymbol{$\cI$}: Last but
      not least, we have not specified precisely how the operator
      $\cI$ is acting. For the moment, we will consider it as an
      abstract operator. This is anyway of great relevance because
      once we are able to recover $L$ (or, alternatively,
      $\mu_L$), we can obtain $\theta$ through
      Equation~\eqref{eq:keyRelation}. Otherwise speaking, we
      could solve SPDE~\eqref{eq:thetaCodebook}.
\end{itemize}

\section{Deep calibration}\label{sec:deepCal}

\subsection{An ill-posed inverse problem}

If we look more closely to steps $(i)-(iii)$ of
    Definition~\ref{def:crcModel}, it is possible to realise that
    the more complex aspect is given by the application of the
    operator $\cI$.  This is basically a calibration conditioned
    on some (new) parameters and state variables whose complexity
    depends on the distribution of the L\'evy process $L$.  In
    general, this is not a trivial operation, since it consists in
    solving an \emph{inverse problem} that is \emph{ill-posed} in
    the sense of Hadamard even for the easiest cases (e.g.~Bates
    model).  The inverse problem is ill-posed because of an
    identifiability issue, which means that the information coming
    from market data is insufficient to exactly identify the
    parameters.  If we express the quantity $\Delta C_t$ of
    Equation~\eqref{eq:calibrCond} as a function of the model
    parameters $\vartheta$, that is
\begin{equation*}
      \Delta C_t(\vartheta) = \sum_{i=1}^n \sum_{j=1}^m \left|C_{t}^{\operatorname{model}}(\vartheta; T_i, K_j) - C_{t}^{\operatorname{obs}}(T_i, K_j)\right|^2,
\end{equation*}
we can write the identifiability problem as the fact that the
    function $\Delta C_t(\vartheta)$ has many local minima.
    Furthermore, it is in general unclear whether these minima can
    be reached by the adopted algorithm.  For example, Cont and
    Tankov show in \cite{ContTankov2004} that if one had available
    a set of call options prices (or, equivalently, implied
    volatilities) for \emph{all} strikes (in a given time
    interval!) and a \emph{single} maturity, then it would be
    possible to deduce all the parameters of the model and, in
    particular, the L\'evy triplet.  But in reality this is never
    the case, since we only know prices for a finite number of
    strikes and, in addition, we also have observational errors in
    the data.  As a result, we have a serious identification
    problem, exemplified by the fact that we can obtain the same
    prices for (infinitely) many combinations of the parameters.
The strategy which they begin to develop in \cite{ContTankov2004} and complete in \cite{ContTankov2005} is the use of the Kullback-Leibler divergence, also called relative entropy, as a regulariser in order to get a well-posed inverse problem.\\
To overcome this issue, we decided to follow a different
    strategy, making use of the implicit regularisation present in
    neural networks.

	\subsection{Learning the inverse map}
	
	Andres Hernandez was among the first who tried neural networks
        (NN) to address calibration tasks in Finance.  In
        \cite{Hernandez1}, he showed that a feedforward NN can
        actually approximate the inverse map given by the pricing
        formula and obtain the two parameters of the Hull-White
        interest rate model $(a,\sigma)$ as output of a NN.  Just as a
        reminder, the Hull-White model consists of the following SDE
	\begin{equation*}
          dr(t) = [\beta(t) - ar(t)]\,dt + \sigma\,dW(t),
	\end{equation*}
	where $a$, $\sigma>0$ and $\beta(t)$ is uniquely determined by
        the term structure\footnote{This curve is calibrated through
          the derivative of the instantaneous forward rate $f(t,T)$ at
          time $t=0$, i.e.
          $\beta(t) = \frac{\partial f(0,t)}{\partial T} + a f(0,t) +
          \frac{\sigma^2}{2a}\left(1 - e^{-2at}\right)$ once the other
          two parameter $a$ and $\sigma$ have been calibrated.}.  The
        greatest achievement that he highlighted in the paper is the
        possibility of replacing the traditional ``slow'' and
        cumbersome calibration procedure with a new straightforward
        deterministic map, which makes calibration itself a very
        efficient task, since the core of all calculations is offset
        to the training phase.  In fact, once a NN is trained, its
        application is extremely cheap from a computational point of
        view, being the most expensive operations simple matrix vector
        multiplications.

	Despite the good results obtained by Hernandez, learning the
        map from the prices to the parameters can be critical, since
        this map is not known in explicit form.  In principle, we do
        not even know if the universal approximation theorem could be
        applied, because the direct map could not be bijective (thus
        having a discontinuous inverse function).  In general, since
        the inverse map is not known, we lack control on it and it
        might well be that a NN learns appropriately the map on the
        given training sample, but is not able to generalise on
        out-of-sample data.  This is actually what happened when we
        tried to apply this approach to our problem since our
        situation is considerably more involved than
        Hernandez'. \hfill \break

	For this reason, it is a better idea to learn the direct (or
        forward) map from the parameters to the prices/implied
        volatilities.  This is done, for example, by Horvath and
        coauthors in \cite{Horvath2019}, where they implemented a
        feedforward NN to directly get the volatilities from the model
        parameters.  Note that for the training of the networks, the
        data are artificially generated and the grid of strikes and
        maturities is fixed at the beginning.  Again, the most
        appealing advantage they see in the application of NN is the
        possibility of enabling live calibration of derivative
        instruments, since the application of the NN itself only
        requires milliseconds avoiding the traditional bottleneck of
        calibration.  This allows making use of new models that were
        before considered too computationally expansive for use, such
        as the rough volatility models (e.g.~rough Heston or rough
        Bergomi model, which require Monte Carlo algorithms).

	Despite the encouraging results found in \cite{Horvath2019},
        there are still some inconveniences using this last approach.
        If on one hand the advantages on the speed side are evident,
        on the other there are still downsides that are relevant, but
        not addressed by this kind of solution.  Indeed, problems
        might come by the second step of this procedure, as denoted in
        \cite{Horvath2019}, which is the real calibration.  Even if we
        have to face a deterministic optimisation problem, this might
        not be as easy as it seems, in particular for multidimensional
        models.
	In these cases, we might need to use a local optimiser to speed up, which usually requires prior knowledge about the solution, or a global optimiser, which might take longer time.\\
	
	To overcome these issues, we propose a new method that allows
        learning the inverse map, as already done by Hernandez, even
        for ill-posed problems.  It is our wish to underline that the
        same idea could be used to learn the inverse map and solve
        inverse problems in fields other than mathematical finance.
        Adopting the same approach as Hernandez did not work out in
        our case, because of identifiability issues: different
        combinations of parameters in the stochastic volatility
        Hull-White extended affine model result in the same volatility
        surface.  Thus, our idea is allowing a neural network to
        decide autonomously which parameters giving as output knowing
        that these parameters will then have to give rise to the
        prescribed volatility surface.  In this way, we will also be
        able to learn the operator $\cI$ defined in
        \eqref{eq:operatorI}.
	
	In order to make the system works, we need two neural
        networks.  The first, denoted in the following as $\nn_1$, is
        a map between parameters and volatilities (basically, the
        usual pricing function, as learnt in \cite{Horvath2019}),
        while the second, called $\nn_2$, maps volatilities to
        parameters (but is not trained in the usual way); in our case,
        to the parameters defining the L\'evy process $L$.  Then, we
        compose the two networks, where the first is trained, while
        the second is not, to obtain a new neural network $\nn_3$
        which receives in input volatilities and returns the (same)
        volatilities:
	\begin{equation*}
          \nn_3 \define \nn_2 \circ \nn_1.
	\end{equation*}
	In other words, $\nn_3$ will learn the identity and, during
        the training phase, $\nn_2$ will get trained.  In this
        respect, we can see $\nn_2$ as the \emph{inverse} neural
        network of $\nn_1$.  The trick is as simple as that.  However,
        notice that we might not recover exactly the same parameters
        that gave birth to the original IVS, but an
        \emph{equivalent}\footnote{We could think of the combinations
          selected by the neural network as the representatives for an
          equivalence class, where all members originate the same
          implied volatility surface.} combination that resulted in
        the same surface through $\nn_1$.
	
	\begin{figure}
          \centering
          \includegraphics[width=1.\textwidth]{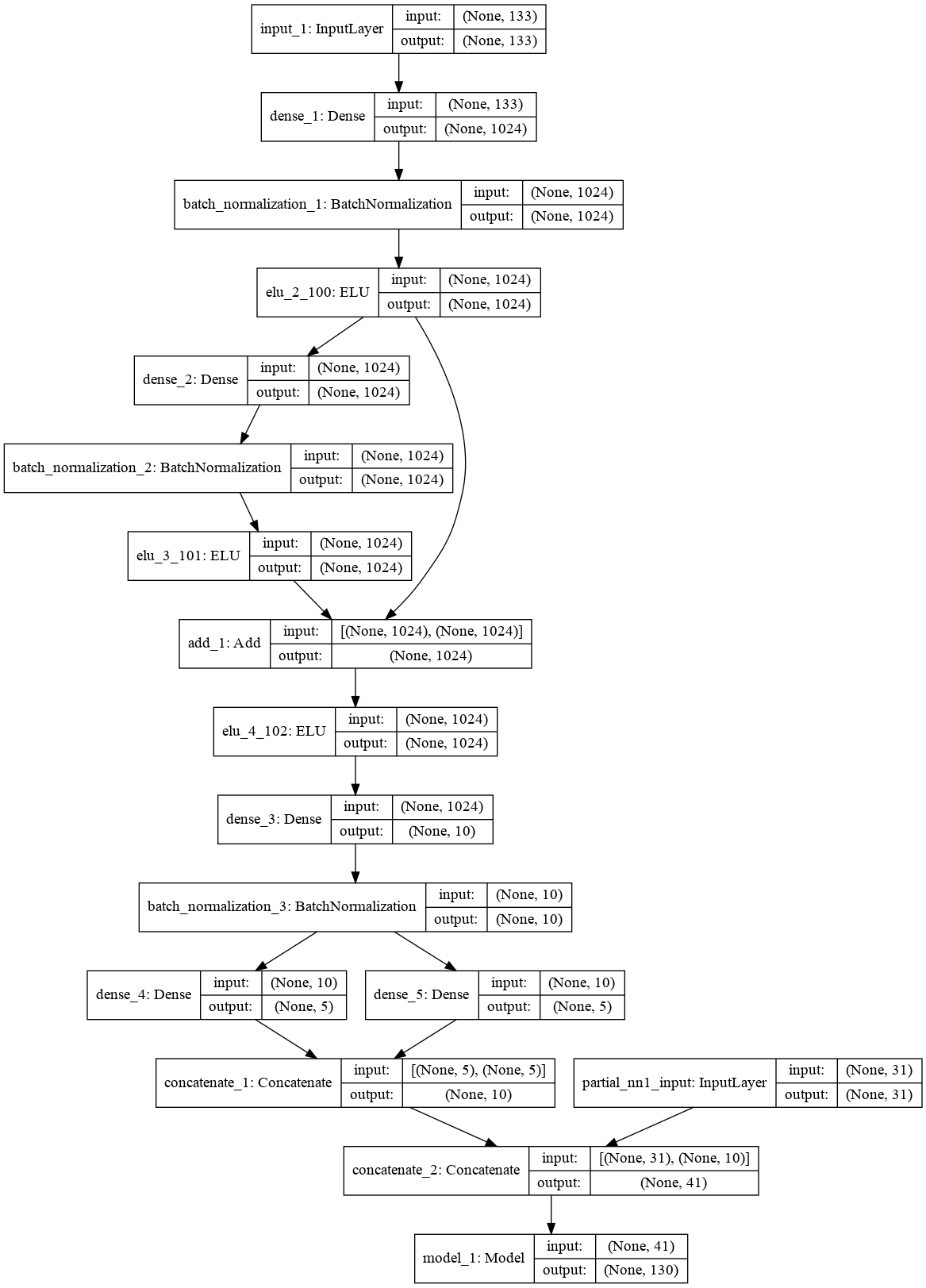}
          \caption{\small Keras representation of $\nn_3$ (with just 2
            ``main'' hidden layers instead of 4 for $\nn_2$ to fit the
            picture in the page). $\nn_1$ is summarised as model\_1 at
            the very bottom. Between the 2 wanted hidden layers elu\_2
            and elu\_4, it is possible to find another layer (in this
            sense the nomenclature \emph{1-cell} that we used in the
            text).}
          \label{fig:kerasNN3}
	\end{figure}

	\subsection{Numerical implementation}
	
	In broad terms, the numerical implementation follows the model
        outlined in Section \ref{sec:HestonBates}, where the process
        $L$ is a compensated compound Poisson process in which the
        size of jumps is normally distributed.  Since the parameters
        are allowed to change in time, the mean and variance of the
        Gaussian distribution for the jump-size of $L$ are
        time-dependent, while the Poisson rate is considered constant
        in time.  The same holds true for the other parameters
        belonging to the ``Heston'' part, with the exception of the
        interest rate $r$, the dividend rate $q$ and $k$, which is the
        speed of mean reversion for the variance process.  We decided
        to free ourselves from a static framework and to have a
        variable maturities-strikes grid.  More precisely, we used ten
        different time-to-maturities $\{\tau_i\}_{i=1}^{10}$, with
        $\tau_1 < \dots < \tau_{10}$ ranging between 7 and 440 days
        (extremes included), being more concentrated for short
        maturities, and thirteen different moneyness
        $m_1 < \dots < m_{13}$ ranging between 0.8 and 1.2 (extremes
        included) in strictly increasing order.  One difference with
        the generalised Bates model described in
        Section~\ref{sec:HestonBates} is that we have a
        maturity-dependent jump distribution: mean and variance depend
        on the maturity in the sense that they are piecewise constant
        within two adjacent time-to-maturities, therefore, the process
        $L$ is here modelled through 11 parameters: the Poisson rate
        and then 5 tuples mean-variance for the normal distributions.
	
	The first neural network, i.e.~$\nn_1$, is a
        1-cell\footnote{With this, we mean that between the predefined
          hidden layers we only find one time the application of the
          activation function (basically, another hidden layer). The
          situation might be clearer by looking at
          Figure~\ref{fig:kerasNN3}.} residual feedforward NN (see
        \cite{ResNet2015} for more information on ResNets) composed by
        4 ``main'' hidden layers with 1024 nodes each. The input layer
        has dimension 41 and includes
	\begin{equation*}
          r, q, \{\tau_i\}_{i=1}^{10}, \{m_i\}_{i=1}^{13}, v_0, k, \theta, \sigma, \rho, \lambda, \{\nu_i\}_{i=1}^{5}, \{\delta_i\}_{i=1}^{5},
	\end{equation*}
	where $\nu_i$ and $\delta_i$ are the mean and standard
        deviations of the normal distributions for the jump size.  The
        output layer has dimension 130 and includes the entire
        point-valued volatility surface, denoted as
        $$ \{\text{IVS}_i\}_{i=1}^{130}.
	$$
	The activation function used for all layers (apart from the
        output layer) is ELU.  This network is trained first with
        artificially generated data: all parameters are sampled from
        uniform distributions whose extremes (parameters) are defined
        a priori (and are kept fixed throughout the process).  Then,
        QuantLib Python routines (see \cite{Ame2003}) are used to
        obtain in an efficient and fast way all the necessary prices.
        Implied volatilities are then retrieved through the algorithm
        outlined by Fabien Le Floc’h in
        \url{http://chasethedevil.github.io/post/implied-volatility-from-black-scholes-price/}
        and implemented in Python.
	
	Second, $\nn_2$ is created, but not (immediately) trained. As
        already explained, this second neural network will be trained
        only after being composed with the trained $\nn_1$, which will
        be marked as \emph{non-trainable} in this second phase.  This
        composed NN is called $\nn_3$.  In order to learn the operator
        $\cI$, $\nn_3$ will be trained and, as a side result, $\nn_2$
        will be also trained.  That is to say that $\nn_3$ is just
        used as a mere tool to arrive to get $\nn_2$ trained as
        well. Finally, it will be then separated from $\nn_1$.  As
        already said, the goal of $\nn_3$ is basically learning the
        identity function.  Thus, the input of $\nn_2$ is the
        following:
	\begin{equation*}
          \theta, \sigma, \rho, \{\text{IVS}_i\}_{i=1}^{130},
	\end{equation*}
	while the output, since we have to learn the L\'evy process
        $L$, is
	\begin{equation*}
          \{\nu_i\}_{i=1}^{5}, \{\delta_i\}_{i=1}^{5}.
	\end{equation*}
	From an architectural viewpoint, $\nn_2$ has 4 ``main'' hidden
        layers with 1024 nodes each.  The activation function is ELU,
        as for $\nn_1$.  While training $\nn_3$ (and, implicitly,
        $\nn_2$), we have to provide as output the entire implied
        volatility structure $\{\text{IVS}_i\}_{i=1}^{130}$, while as
        input the concatenation of the complete input of $\nn_2$, plus
        the incomplete input of $\nn_1$, that is everything listed
        above apart from
        $\{\nu_i\}_{i=1}^{5}, \{\delta_i\}_{i=1}^{5}$, which have to
        be guessed during the training.  To obtain a satisfactory
        training procedure, we tried also different activation
        functions for the output layer of $\nn_2$.  In the end, the
        best results were reached using the standard sigmoid function
        stretched to completely cover the intervals\footnote{Instead
          of the interval $[0,1]$ which represents the codomain of the
          function.} in which $\nu_i$ and $\delta_i$ were (randomly)
        extracted.
	Without this precaution the training process could not converge to a reliable result.\\
	For both training processes, we used the mean squared error on
        the implied volatilities as loss function, since we were
        dealing with regression-type tasks (other loss functions were
        tried, but they gave birth to NNs that were operating more
        poorly).  To obtain better results, it was really helpful also
        the linear transformation operated on the input and output
        data: outside of the implied volatilities which were kept
        unchanged, all other quantities were scaled to reside in the
        interval $(0,1)$.  Moreover, as it is possible to see from
        Figure~\ref{fig:kerasNN3}, we made use of batch normalisation,
        while we avoided drop-out.  The best batch size for both
        training processes was 1'000 (out of a database made of
        around 600'000 elements).
	All hyperparameters have been selected after tuning the networks, using not only manual adjustments, but also other techniques like \emph{randomised search}.\\
	The whole neural network architecture was developed using
    Keras. A schematic representation can be found in
    Figure~\ref{figure:NN}.
    
    \medskip
    The interested reader who would like to grasp the approximations
    capabili\-ties of the implemented neural networks is addressed to
    Appendix \ref{sec:plotsNNs}.

{\centering

\begin{sidewaysfigure}[p]
      \tikzset{%
        every neuron/.style={ circle, draw, minimum
          size=0.8cm }, neuron missing/.style={ draw=none,
          scale=2, text height=0.333cm, execute at begin
          node=\color{black}$\vdots$ } }

      \begin{tikzpicture}[x=1.5cm, y=1.5cm, >=stealth]
        \foreach \m/\l [count=\y] in {1,2,3,4,missing,5}
        \node [every neuron/.try, neuron \m/.try]
        (input1-\m) at (0,2.5-\y) {};

        \foreach \m [count=\y] in {1,missing,2} \node
        [every neuron/.try, neuron \m/.try ] (hidden1-\m)
        at (1.2,1.5-\y*1.2) {};

        \foreach \m [count=\y] in {1,missing,2} \node
        [every neuron/.try, neuron \m/.try ] (hidden2-\m)
        at (2.4,1.5-\y*1.2) {};

        \foreach \l [count=\i] in {1,2,3,4,n} \draw [<-]
        (input1-\i) -- ++(-1,0) node [above, midway] {};

        \foreach \l [count=\i] in {1,n} \draw [->]
        (hidden2-\i) -- ++(1,0) node [above, midway] {};

        \foreach \i in {1,...,5} \foreach \j in {1,...,2}
        \draw [->] (input1-\i) -- (hidden1-\j);

        \foreach \i in {1,...,2} \foreach \j in {1,...,2}
        \draw [->] (hidden1-\i) -- (hidden2-\j);

        \node [neuron missing/.try] (miss1) at (3.6,
        1.5-2*1.2) {};

        \foreach \m/\l [count=\y] in {3,missing,4} \node
        [every neuron/.try, neuron \m/.try] (output1-\m)
        at (4.8,1.5-\y*1.2) {};

        \foreach \l [count=\i] in {3,4} \draw [<-]
        (output1-\l) -- ++(-1,0) node [above, midway] {};

        \foreach \m/\l [count=\y] in {1,2,missing} \node
        [every neuron/.try, neuron \m/.try] (input2-\m) at
        (4.8,5.1-\y*1.2) {};

        \foreach \l [count=\i] in {1,2} \draw [<-]
        (input2-\i) -- ++(-1,0) node [above, midway] {};

        \foreach \m [count=\y] in {1,2,missing,3,4} \node
        [every neuron/.try, neuron \m/.try ] (hidden3-\m)
        at (6,4.5-\y*1.2) {};

        \node [neuron missing/.try] (miss2) at (7.2,
        4.5-3*1.2) {};

        \foreach \i in {3,...,4} \foreach \j in {1,...,4}
        \draw [->] (output1-\i) -- (hidden3-\j); \foreach
        \i in {1,...,2} \foreach \j in {1,...,4} \draw
        [->] (input2-\i) -- (hidden3-\j);

        \foreach \l [count=\i] in {1,2,3,n} \draw [->]
        (hidden3-\i) -- ++(1,0) node [above, midway] {};

        \foreach \m/\l [count=\y] in {1,2,missing,3,4}
        \node [every neuron/.try, neuron \m/.try]
        (hidden4-\m) at (8.4,4.5-\y*1.2) {};

        \foreach \l [count=\i] in {1,2,3,4} \draw [<-]
        (hidden4-\l) -- ++(-1,0) node [above, midway] {};

        \foreach \m/\l [count=\y] in {1,2,missing,3} \node
        [every neuron/.try, neuron \m/.try] (output2-\m)
        at (9.6,4.-\y*1.2) {};

        \foreach \i in {1,...,4} \foreach \j in {1,...,3}
        \draw [->] (hidden4-\i) -- (output2-\j);

        \foreach \l [count=\i] in {1,2} \draw [->]
        (output2-\i) -- ++(1,0) node [above, midway]
        {$\text{IVS}_\l$}; \draw [->] (output2-3) --
        ++(1,0) node [above, midway] {$\text{IVS}_{130}$};

        \draw
        [decorate,decoration={brace,amplitude=10pt},xshift=4pt,yshift=0pt]
        (3.5,1.) -- (3.5,4.0) node
        [black, midway, align=center, xshift=-1.cm]
        {\footnotesize Partial\\\footnotesize input
          $\nn_1$};

        \draw[decoration={brace,mirror,raise=2pt,amplitude=10pt},decorate]
        (-0.5,-4) -- node[below=18pt] {$\nn_2$} (5.2,-4);

        \draw[decoration={brace,raise=-2pt,amplitude=10pt},decorate]
        (4.4,4.5) -- node[above=14pt] {$\nn_1$} (10.,4.5);

        \node[] at (4.2,4.1) {$r$}; \node[] at (4.2,2.9)
        {$q$}; \node[] at (4.2,0.5) {$\nu_1$}; \node[] at
        (4.2,-1.9) {$\delta_5$}; \node[] at (-0.6,1.7)
        {$\theta$}; \node[] at (-0.6,0.7) {$\sigma$};
        \node[] at (-0.6,-0.3) {$\rho$}; \node[] at
        (-0.6,-1.3) {$\text{IVS}_{1}$}; \node[] at
        (-0.6,-3.3) {$\text{IVS}_{130}$};

      \end{tikzpicture}
      \caption{Representation of $\nn_3$ as a
        fully-connected feedforward neural network
        (residual cells are ignored).}
      \label{figure:NN}
\end{sidewaysfigure}
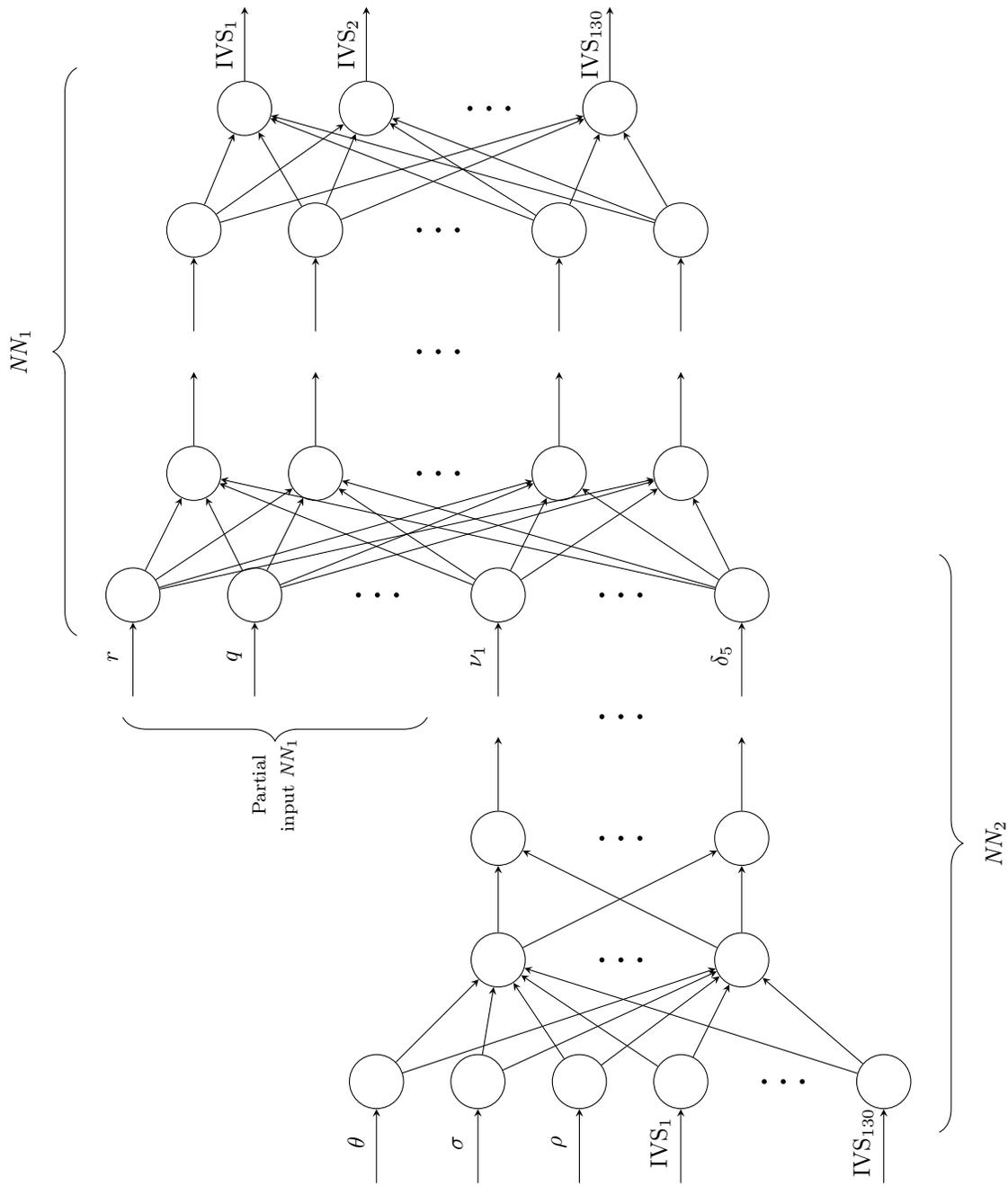
  }

  \subsection{A side result: moving IVS}

  Having a concrete numerical tool that allows the
  recalibration of our model online and basically in an
  instantaneous way has another ``cheerful'' consequence.
  Let us imagine, for the moment, that the dynamics of the
  parameters $p$ are known.  If this is the case, then we
  can model for indefinite time the evolution of an
  implied volatility surface without breaking any
  arbitrage constraints, neither static nor dynamic ones.
  To our knowledge, it is the first time that this
  achieved in an efficient way, one impressive other
  implementation has been presented in
  \cite{CarMaNad:2017}. The algorithm used to accomplish
  that is outlined in Algorithm~\ref{alg:CRCM}.

  \begin{algorithm}
\caption{}\label{alg:CRCM}
\begin{algorithmic}[1]
      \State Pick initial values for the state variables
      (return of asset $X$ and variance $V$), the Heston
      parameters ($\theta, \sigma,\rho$ and $k$), the
      jump-frequency
      $\sim \operatorname{Pois} \left({\lambda\,
          dt}\right)$ and the jump-size normal
      distribution
      $\sim \mathcal{N}(\nu_{i}, \delta_{i}^2)$ for
      $i=1,\dots, 5$. Parameters $\lambda$ and $\kappa$
      remain fixed throughout the procedure.  \State
      Compute the implied volatility surface (IVS) given
      the initial values.  \State \textbf{Bates step}:
      update the two state variables $X$ and $V$ in
      $X_{\text{new}}$ and $V_{\text{new}}$.  \State
      Compute the new implied volatility surface
      IVS\textsubscript{new} given $X_{\text{new}}$ and
      $V_{\text{new}}$.  \State \textbf{Heston-parameter
        step}: update the three parameters $\theta$,
      $\sigma$, $\rho$ according to an exogenously given
      dynamics and obtain $\theta_{\text{new}}$,
      $\sigma_{\text{new}}$, $\rho_{\text{new}}$.  \State
      Given IVS\textsubscript{new} together with
      $\theta_{\text{new}}$, $\sigma_{\text{new}}$,
      $\rho_{\text{new}}$, compute the new parameters
      $(\nu_i^{\text{new}},
      \delta_i^{\text{new}})_{i=1}^5$ such that the IVS
      obtained with $X_{\text{new}}$, $V_{\text{new}}$,
      $\rho_{\text{new}}$, $\theta_{\text{new}}$,
      $\sigma_{\text{new}}$, $\lambda$, $k$ and
      $(\nu_i^{\text{new}},
      \delta_i^{\text{new}})_{i=1}^5$ remains constant
      (equal to IVS\textsubscript{new}).  \State Overwrite
      the initial parameters with the new parameters
      (having the sub/super-script \emph{new}).  \State
      Restart from point 3.
\end{algorithmic}
  \end{algorithm}

  \bigskip

  As already written in the algorithm and for our
  purposes, we decided to initially pick randomly the
  parameters $\theta$, $\sigma$, $\rho$, but then letting
  them evolve according to very simple dynamics, namely
  adding some noise to the current value to get the new
  one.  The variance of the Gaussian noise has been chosen
  relatively small and values are scaled if they overcome
  a certain threshold, so that the relative change (with
  respect to the initial value) could not exceed 5\%.  In
  addition, we made sure that the Feller condition was
  always satisfied and that the values could not exit the
  natural domains we assigned them.  For example, if the
  correlation $\rho$ were brought outside of the interval
  $[-1,1]$, then we would force it to remain inside by
  collapsing the value to the closest extreme. For both
  $\theta$ and $\sigma$ the interval $[0.01, 0.5]$ was
  chosen.

  Notice also that Steps 2, 4 and 6 of
  Algorithm~\ref{alg:CRCM} are made by neural networks,
  $\nn_1$ for 2 and 4, while $\nn_2$ for Step 6.

  Solving the same problem with the desired precision
  without neural networks would have required an immense
  computational power, since the inverse problem is
  notably ill-posed and the regularised inverse problem
  has to be solved at any point in time along the
  discretisation grid.  This is something possible on a
  standard laptop only through these techniques.  Finally,
  it is important to underline that we do not break any
  arbitrage condition because the CNKK drift condition is
  fully incorporated in the steps of
  Algorithm~\ref{alg:CRCM}.

\section{Conclusion}

In this paper, we tried to set up a new rigorous framework in
continuous time for the dynamics of volatility surfaces (or
cubes, etc\dots), so called consistent recalibration models
with applications.  To do so, we took inspiration from similar
work in discrete time by Richter and Teichmann
\cite{Richter2017} and another paper by Harms \emph{et al.}
\cite{paperYieldCurve2018}, which builds the theory in
continuous time, but focusing on yield curves modelling.  With
respect to the latter, in our case we have a more complex
setting due to the more complex term structure, which is here
enriched with a ``strike'' dimension.  This is reflected in
what we called CNKK equation, a generalisation of the more
popular HJM equation, but with considerably more involved
drift term.  It goes without saying that this made the
equation intractable from an analytical point of view.

To overcome this issue, we decided to represent the drift term
by neural networks. We therefore proposed a new way of solving
the (ill-posed) calibration problem, by exploiting the fact
that composition of neural networks is still a neural network
and by defining, in this sense, a sort of \emph{inverse}
network applying implicit regularisation.  The same trick
could be used for other applications, also in branches other
than mathematical finance, to solve inverse problems.  The use
of neural networks was crucial to make numerical procedures
tractable and to get information on the solution of the CNKK
SPDE.  In this case, we can say that the neural network helped
us solving an equation which we could not even write down
(explicitly).

Finally, we could use the same inverse neural network to
simulate the evolution in time of an implied volatility
surface, in this case generated by a generalised Bates model.
To the best of our knowledge, it is the first time this can be
achieved for indefinite time without breaking arbitrage
constraints. In this way, we are even implicitly building 
a realistic equity option market simulator capable of avoiding
any form of arbitrage.

\appendix
\section{Graphical results}\label{sec:plotsNNs}

In this appendix, we report some of the pictures produced using the model outlined in Section \ref{sec:CRC_math} with Python
and the graphical package Matplotlib\footnote{J. D. Hunter, "Matplotlib: A 2D Graphics Environment", Computing in Science \& Engineering, vol. 9, no. 3, pp. 90-95, 2007}.\\
In the first case, plotted figures represent a 3D representation of implied
volatility surfaces together with a heat-map reporting the (pointwise)
differences between the original implied volatility surface (Original IVS) and the
one obtained by application of neural network $\nn_1$ (New IVS), which takes
parameters in input as return the IVS on a grid given by 13 moneyness (or
strikes) and 10 maturities. 
The heat-map was produced using the command \textsf{pcolormesh}.
All (input) parameters were randomly generated from a uniform distribution,
the same used for the generation of training data (but of course not used
during the training process). For sake of simplicity, here we calibrated
just one couple $(\nu, \delta)$ for each volatility surface.

\begin{figure}[H]
  \centering
  \includegraphics[width=1.0\textwidth]{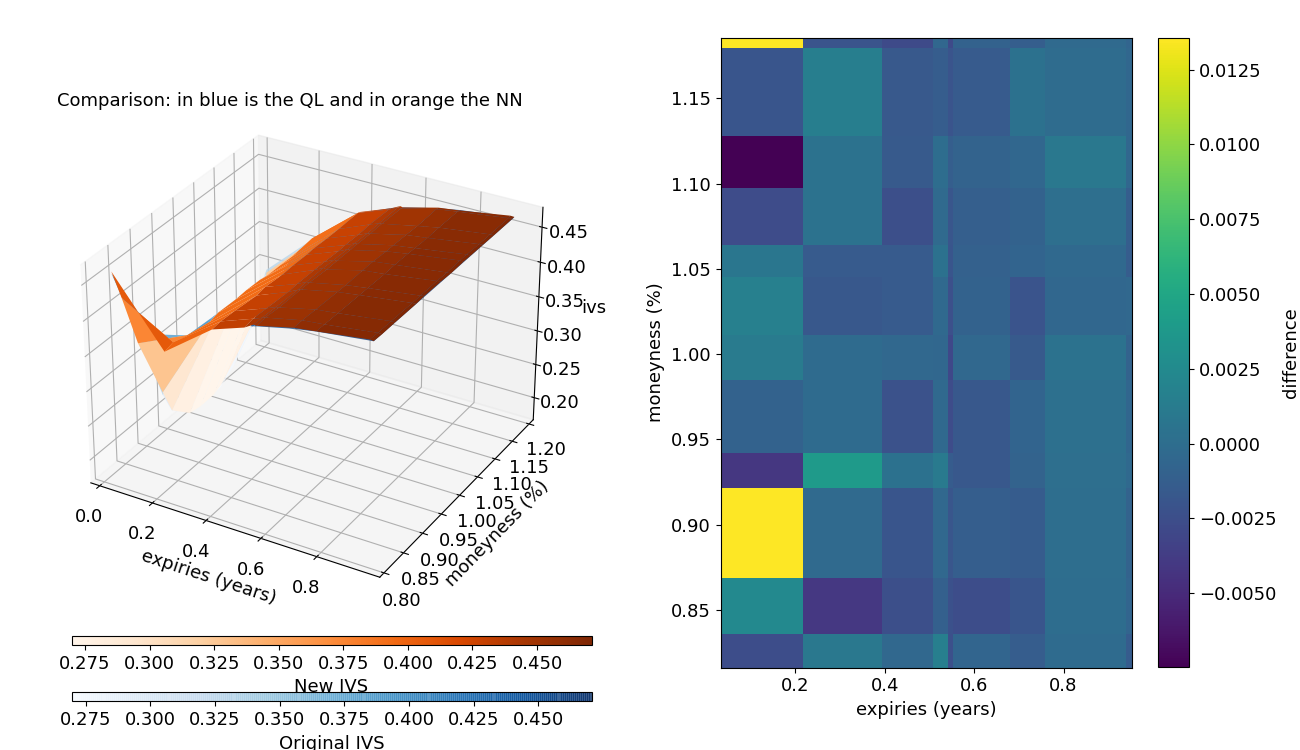}
  \caption{\small Parameters: $S_0~=~$100, $r$~=~0.0205, $q$~=~0.03, $V_0$~=~0.0001, $\kappa$~=~7.797, $\theta$~=~0.247, $\sigma$~=~0.280, $\rho$~=~0.042, $\lambda$~=~0.081, $\nu$~=~0.159, $\delta$~=~0.205}
  \label{fig:IVS1}
\end{figure}

\begin{figure}[H]
  \centering
  \includegraphics[width=1.0\textwidth]{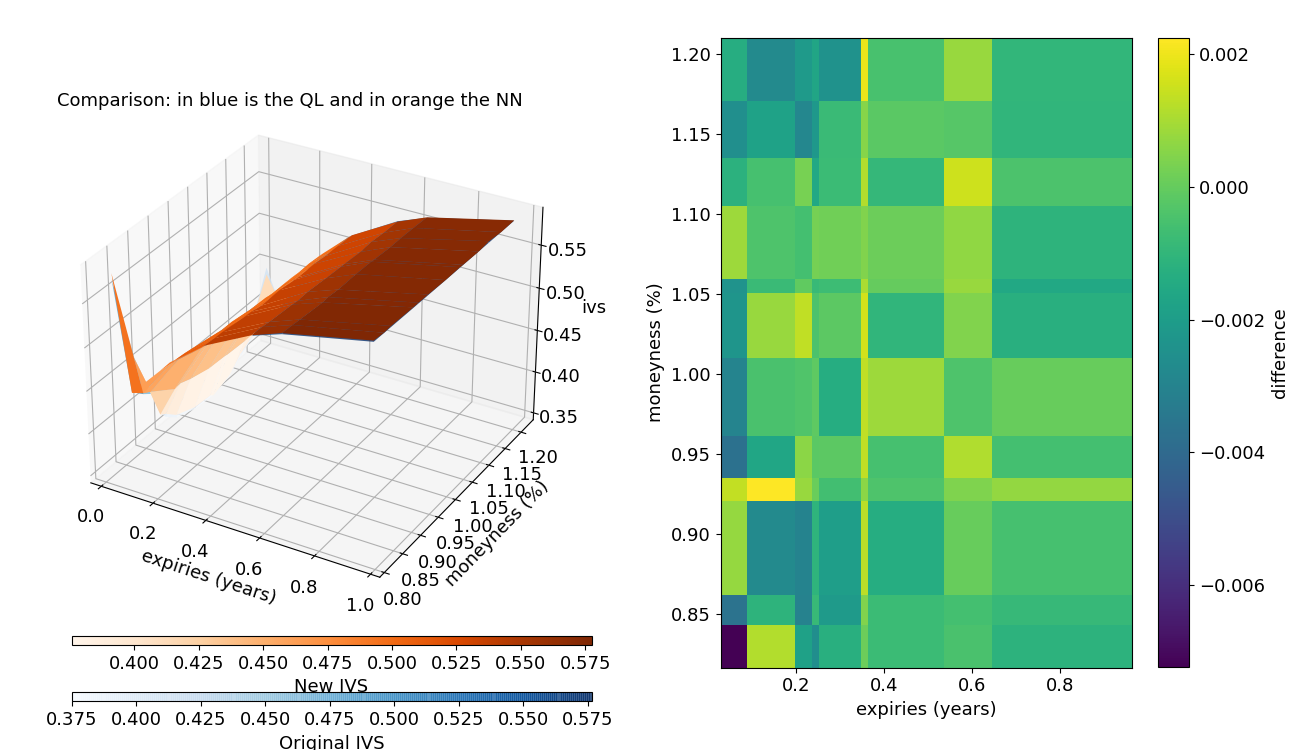}
  \caption{\small Parameters: $S_0$~=~100, $r$~=~0.0068, $q$~=~0.0161, $V_0$~=~0.0951, $\kappa$~=~5.421, $\theta$~=~0.370, $\sigma$~=~0.224, $\rho$~=~0.242, $\lambda$~=~0.289, $\nu$~=~0.087, $\delta$~=~0.249}
  \label{fig:IVS2}
\end{figure}

\begin{figure}[H]
  \centering
  \includegraphics[width=1.0\textwidth]{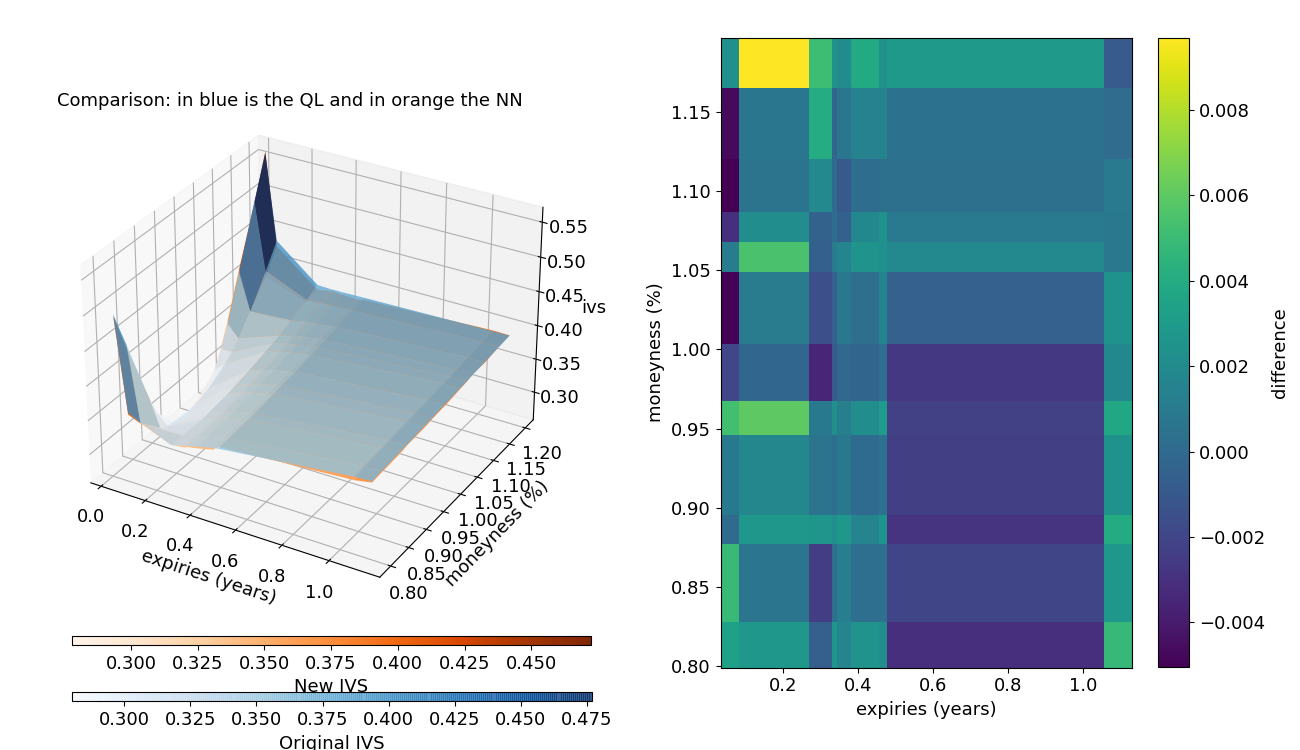}
  \caption{\small Parameters: $S_0$~=~100, $r$ = 0.0111, $q$ = 0.0021, $V_0$~=~0.0552$, \kappa$~=~8.698, $\theta$~=~0.106, $\sigma$~=~0.391, $\rho$~=~-0.12, $\lambda$~=~0.491, $\nu$~=~-0.202, $\delta$~=~0.287}
  \label{fig:IVS3}
\end{figure}

\newpage

On the other hand, in the next figures we will take into consideration
the neural network we called in Section \ref{sec:deepCal} $\nn_2$.
Figures \ref{fig:IVS1_NN2}, \ref{fig:IVS2_NN2} and \ref{fig:IVS3_NN2} were basically obtained after one loop of Algorithm \ref{alg:CRCM}, in the sense that we started from an IVS generated by a model in which the price process and the variance process evolved in time (step 3 in the algorithm), we let the 3 parameters $\theta$, $\sigma$ and $\rho$ changing according to exogenous dynamics (essentially, normal noise) and then we exploited $\nn_2$ to recover the jump parameters that would give rise to the same IVS (IVS\textsubscript{new} in the same algorithm).
Note that the `Original IVS' (in the plots) are obtained for Figures \ref{fig:IVS1_NN2} and \ref{fig:IVS2_NN2} in an analytical way, while in Figure \ref{fig:IVS3_NN2} the `Original IVS' is the output of $\nn_1$. The fact that the error is zero everywhere, although the starting and derived (by $\nn_2$) parameters are not the same, means that the inversion of the neural network is indeed effective.

\bigskip
\bigskip

\begin{figure}[H]
  \centering
  \includegraphics[width=1.0\textwidth]{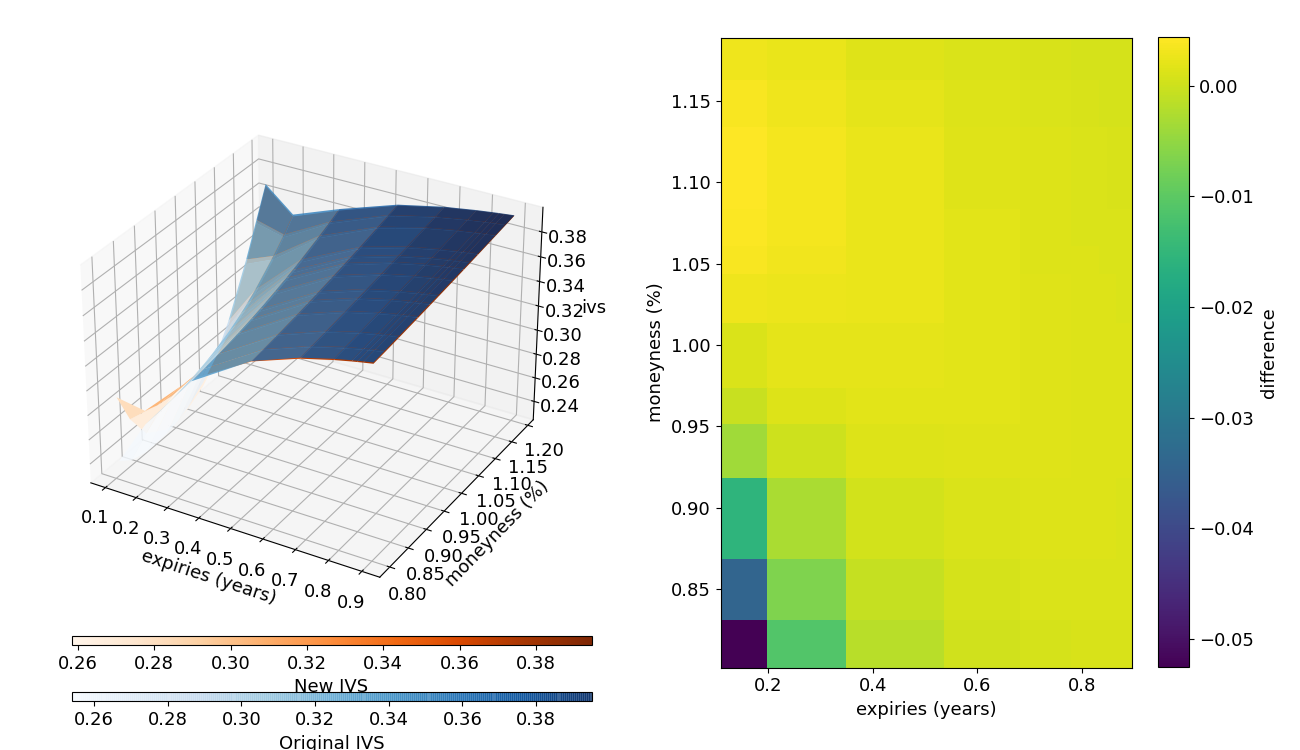}
  \caption{\small Parameters: $S_0$~=~99.783, $r$ = 0.0521, $q$ = 0.0082, $V_0$~=~0.0037, $\kappa$~=~6.924, $\theta$~=~0.146, $\sigma$~=~0.328, $\rho$~=~-0.08, $\lambda$~=~0.295,
  $\nu$~=~-0.286, $\delta$~=~0.211;
 after 1 loop and using $\nn_2$ $\theta$~=~0.142, $\sigma$~=~0.339, $\rho$~=~-0.08, $\nu$~=~-0.238, $\delta$~=~0.299}
  \label{fig:IVS1_NN2}
\end{figure}

\begin{figure}[H]
  \centering
  \includegraphics[width=0.98\textwidth]{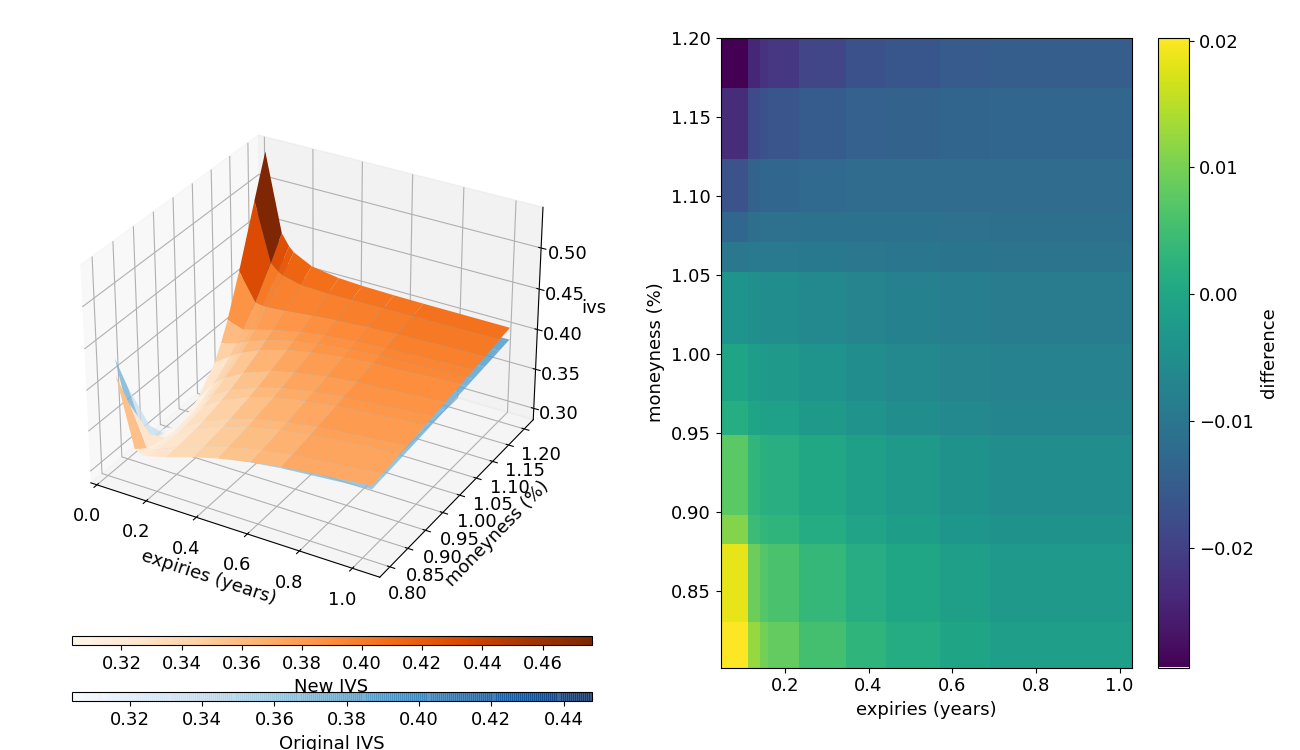}
  \caption{\small Parameters: $S_0$~=~100.26, $r$ = 0.0111, $q$ = 0.0021, $V_0$~=~0.0663, $\kappa$~=~8.698, $\theta$~=~0.106, $\sigma$~=~0.391, $\rho$~=~-0.12, $\lambda$~=~0.491, $\nu$~=~-0.202, $\delta$~=~0.287;
 after 1 loop and using $\nn_2$ $\theta$~=~0.102, $\sigma$~=~0.408, $\rho$~=~-0.12, $\nu$~=~-0.279,  $\delta$~=~0.300.}
  \label{fig:IVS2_NN2}
\end{figure}

\begin{figure}[H]
  \centering
  \includegraphics[width=0.98\textwidth]{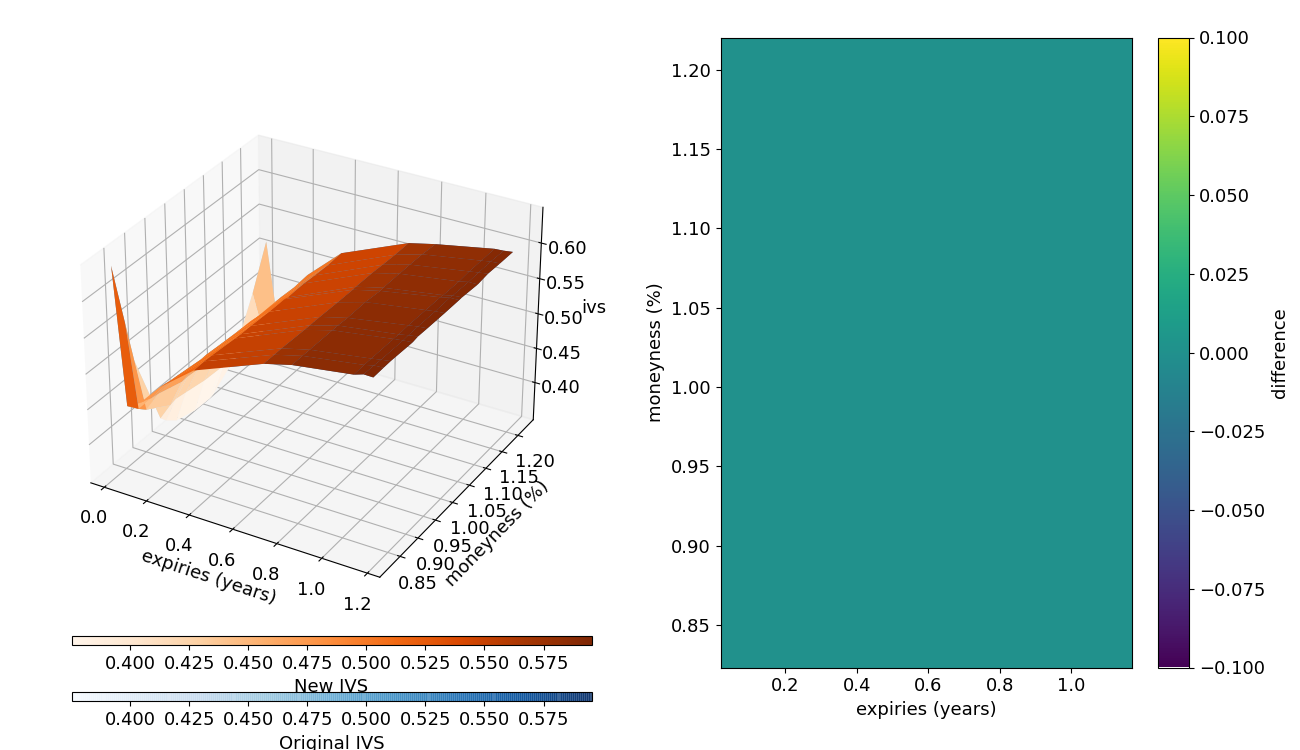}
  \caption{\small Parameters: $S_0$~=~100.83, $r$ = 0.0068, $q$ = 0.0161, $V_0$~=~0.1046, $\kappa$~=~5.421, $\theta$~=~0.370, $\sigma$~=~0.224, $\rho$~=~0.242, $\lambda$~=~0.289, $\nu$~=~0.087, $\delta$~=~0.249;
  after 1 loop and using $\nn_2$ $\theta$~=~0.357, $\sigma$~=~0.218, $\rho$~=~0.251, $\nu$~=~0.289, $\delta$~=~0.300.}
  \label{fig:IVS3_NN2}
\end{figure}

\end{document}